\documentclass[11pt]{article}

\usepackage[margin=1in]{geometry}

\usepackage{amsmath,amssymb,amsthm,mathtools}
\usepackage{bm}
\usepackage{mathrsfs}
\usepackage{bbm}
\usepackage{microtype}
\usepackage{authblk}
\usepackage{natbib}
\usepackage[colorlinks=true,citecolor=blue,linkcolor=blue,urlcolor=blue]{hyperref}

\usepackage{caption}
\usepackage{booktabs}
\usepackage{graphicx}
\usepackage{xcolor}

\newenvironment{fitwide}{}{}

\theoremstyle{plain}
\newtheorem{theorem}{Theorem}
\newtheorem{proposition}{Proposition}

\newtheorem{corollary}{Corollary}

\theoremstyle{definition}
\newtheorem{definition}{Definition}
\newtheorem{assumption}{Assumption}

\newcommand{\R}{\mathbb{R}}

\newcommand{\Nzero}{\mathbb{N}_0}
\newcommand{\p}{\mathbb{P}}
\newcommand{\E}{\mathbb{E}}
\newcommand{\Var}{\mathrm{Var}}
\newcommand{\Cov}{\mathrm{Cov}}
\newcommand{\ind}{\mathbbm{1}}

\newcommand{\FN}{F_N}

\newcommand{\pN}{p_N}
\newcommand{\FX}{F_X}
\newcommand{\fX}{f_X}

\newcommand{\CL}{\mathrm{CL}}
\newcommand{\plim}{\overset{p}{\longrightarrow}}
\newcommand{\dto}{\overset{d}{\longrightarrow}}

\newcommand{\Hmat}{\mathcal{H}}
\newcommand{\Jmat}{\mathcal{J}}

\newcommand{\fullparam}{\bm{\vartheta}}
\renewcommand{\d}{\,\mathrm{d}}

\title{Severity estimation in dependent collective risk models}

\author{Christopher Blier-Wong}
\affil{Department of Statistical Sciences, University of Toronto, Canada}

\date{\today}

\begin{document}
\maketitle

\begin{abstract}
The collective risk model represents the aggregate loss of an insurance portfolio as a random sum of individual claim severities. When claim counts and severities are dependent, the claims pooled across policies are no longer a sample from the marginal severity distribution. We show that their empirical distribution converges to the law of an arbitrary observed claim, a size-biased mixture of the conditional severity distributions, so any procedure that fits the severity margin directly to pooled claims is inconsistent in general. The same result identifies the distribution that the pooled claims do sample, and we build a composite likelihood estimation procedure on that distribution. We establish consistency and asymptotic normality, with Godambe information in which the policy, rather than the claim, is the sampling unit. In a Sarmanov collective risk model, the observed-claim density and the aggregate mean are in closed form. A simulation study measures the bias of naive pooled-severity fitting, its correction by the composite likelihood, and the coverage of the policy-level standard errors.
\end{abstract}

\bigskip


\section{Introduction}\label{sec:intro}

The collective risk model (CRM) is the standard representation of aggregate losses in non-life insurance and risk theory. Over a fixed horizon, say a policy-year, the aggregate loss is
$$S = \sum_{j=1}^{N} X_j,$$
with $S=0$ when $N=0$, where $N\in\Nzero$ is the claim count and $(X_j)_{j\ge 1}$ are claim severities taking values in $\R_+=[0,\infty)$. Pricing, reserving, capital allocation, and solvency assessment all rest on the distribution of $S$ and on risk measures such as its mean, variance, value-at-risk, and expected shortfall. Textbook treatments appear in \citet{KlugmanPanjerWillmot2012} and \citet{KaasGoovaertsDhaeneDenuit2008}, and \citet{DenuitDhaeneGoovaertsKaas2005} survey dependence concepts in actuarial science. This paper concerns the estimation of CRM parameters from portfolio data.

The classical CRM assumes independence: $N$ is independent of the severity sequence, and given $N=n$ the severities $(X_1,\dots,X_n)$ are i.i.d.\ with common distribution $\FX$. Independence makes estimation modular: one fits the count distribution $\FN$ to the claim counts, the severity distribution $\FX$ to the observed claims, and combines the two into a model for $S$. Real portfolios violate both assumptions. Claim frequency and severity are typically negatively dependent, and severities within a policy tend to be positively dependent among themselves \citep{GarridoGenestSchulz2016,shi2020regression}. Under dependence, the severity distribution may vary with the realized count, and the claims within a policy-year may be mutually dependent. Either feature invalidates the modular two-step procedure.

Table~\ref{tab:example} shows the data structure and illustrates the estimation problem. The five rows are simulated from a CRM in which frequency and severity are positively dependent. Policies~1 and~4 report no claims and contribute nothing to the pool of observed severities; policy~2 contributes one claim, and policies~3 and~5 contribute two claims each. The standard workflow pools the five claim amounts and treats them as an i.i.d.\ sample from $\FX$. Under dependence, this pooling introduces a systematic bias. Policies with many claims are over-represented in the pool, and when high-count policies also produce larger claims, the pooled sample shifts away from $\FX$.

\begin{table}[ht]
\centering
\caption{Five policy-year observations from a collective risk model. Each policy reveals its claim count $N_i$ and, when $N_i\ge 1$, the individual claim amounts.}
\label{tab:example}
\begin{tabular}{ccccc}
\toprule
Policy $i$ & $N_i$ & $X_{i1}$ & $X_{i2}$ & $S_i$ \\
\midrule
1 & 0 & -- & -- & 0 \\
2 & 1 & 17.96 & -- & 17.96 \\
3 & 2 & 18.50 & 94.51 & 113.01 \\
4 & 0 & -- & -- & 0 \\
5 & 2 & 67.71 & 94.44 & 162.15 \\
\bottomrule
\end{tabular}
\end{table}

Figure~\ref{fig:sizebias} shows the shift on a portfolio of one million simulated policies. The counts are Poisson with mean $1/2$, the severities are exponential with mean $100$, and the count and each individual severity are positively dependent, with Spearman correlation $0.24$; Section~\ref{sec:sarmanovmodel} details the construction. The empirical distribution of the pooled claims lies well below the true severity distribution $\FX$ over a wide range: pooled claims are stochastically larger than the marginal severity law predicts. A practitioner fitting $\FX$ to these pooled claims would overestimate the severity mean by roughly $25\%$. The dashed blue curve is the distribution of an \emph{arbitrary observed claim} implied by the model, the central object of this paper. It tracks the pooled empirical distribution exactly.

\begin{figure}[ht]
\centering
\includegraphics[width=0.6\textwidth]{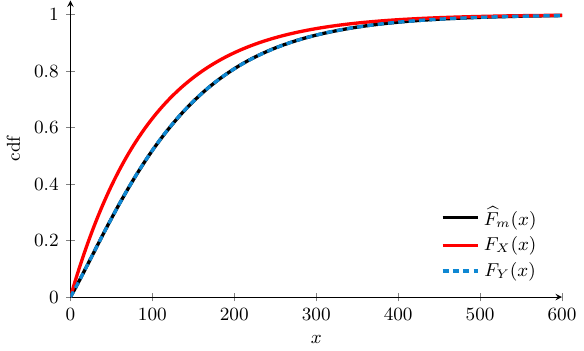}
\caption{Empirical cdf $\widehat{F}_m$ of the pooled claim severities (black), true marginal severity cdf $\FX$ (red), and observed-claim cdf $F_Y$ (dashed blue) for $m=1\,000\,000$ policies simulated from a CRM with Poisson$(1/2)$ counts, exponential severities with mean $100$, and positive frequency--severity dependence from an FGM copula with $\theta_{01}=0.72$, generated from a projectively coherent construction with within-policy severity dependence. The pooled claims track $F_Y$, not $\FX$.}
\label{fig:sizebias}
\end{figure}

The phenomenon in Figure~\ref{fig:sizebias} is a form of size-biased sampling: a policy that reports more claims contributes more of them to the pool, so the pool over-weights the severity behaviour of high-count policies. Our first result (Theorem~\ref{thm:FYlimit}) makes this precise. The empirical measure of the pooled claims converges almost surely to the law of an arbitrary observed claim $Y$, whose cdf is the size-biased mixture
$$F_Y(x)=\sum_{n=1}^{\infty}\frac{n\,\p(N=n)}{\E[N]}\,F_{X\mid N=n}(x),\qquad x\ge 0,$$
of the conditional severity laws. In the independent CRM, $F_{X\mid N=n}=\FX$ for every $n\ge 1$, the mixture collapses to $\FX$, so fitting the severity margin to the pooled claims recovers the true severity distribution. Under dependence, $F_Y\ne\FX$ in general. We call an approach margin-first when it proceeds in two stages: it first fits the severity margin to the uncorrected pooled claims, and then estimates the dependence parameters with the margin held fixed. Maximum likelihood on the pooled claims, inference functions for margins (IFM), and semiparametric rank-based copula methods all follow this pattern. Any such procedure converges to the parameter of the severity distribution closest to $F_Y$ in the Kullback--Leibler sense, rather than to the true severity parameter (Corollary~\ref{cor:naiveMLE}). 

Having identified the data-generating mechanism, we can also use it to design a consistent estimation procedure. The pooled claims sample exactly the law $F_Y$, so whenever the model implies a tractable form for $F_Y$, they provide a valid likelihood component for the severity and dependence parameters. We combine this component with the count likelihood into a composite likelihood, that is, an estimation criterion assembled from low-dimensional components of the model rather than from the full joint density \citep{Lindsay1988,VarinReidFirth2011}, and we prove consistency and asymptotic normality of the resulting estimator (Theorem~\ref{thm:CLasymp}). Each policy is the sampling unit, and it contributes a whole group of claims at once: a random number of them, determined by its claim count, and mutually dependent, because they are generated under the same policy. We call this group the policy's cluster of claims. Because the claims within a cluster are dependent, the usual equality between the variance of the score and the expected curvature of the criterion fails. The asymptotic variance is instead the inverse Godambe information, a sandwich of the score variance and the expected curvature \citep{Godambe1960,GodambeHeyde1987}, and its validity requires that scores be aggregated at the policy level. Aggregating at the claim level understates the variance whenever within-policy severities are dependent. A stepwise extension estimates higher-order dependence parameters from low-dimensional conditional margins, with a stacked sandwich that propagates the plug-in uncertainty of the earlier steps.

The composite likelihood needs a model in which $F_Y$ is explicit, and the Sarmanov CRM is one such case. In this model, a Sarmanov copula \citep{Sarmanov1966} generates the dependence between $N$ and the severity sequence in every finite dimension. The family generalizes the FGM copula and admits a latent Bernoulli-mixture representation \citep{blier-wong2026stochastic} that yields closed-form low-order margins, thereby yielding a closed-form observed-claim density $f_Y$, together with admissibility constraints expressed through the nonnegativity of a Bernoulli probability mass function. We illustrate the procedure through a simulation study: naive pooled fitting carries a bias that does not vanish with the sample size, the composite likelihood removes it, and the policy-level sandwich attains nominal coverage while its claim-level counterpart undercovers.

This paper sits at the intersection of three strands of literature in statistics and actuarial science. The first is copula-based CRMs. Early models captured dependence between the claim count and the average severity per policy; see \citet{FreesValdez2008,CzadoEtAl2012,KramerEtAl2013,ShiFengIvantsova2015,GarridoGenestSchulz2016}. Special cases with Sarmanov distributions appear in \citet{BolanceVernic2020,VernicBolanceAlemany2022}. Modelling the average severity is adequate when the target is the mean of $S$. It is not enough for the distribution of $S$ or of the individual claims, which requires the law of each severity rather than a single number per policy. Estimation in that class uses one severity value per policy, typically the average or total claim amount, so the pooling problem studied here does not arise. More recent constructions allow dependence between the count and the individual severities, as well as among the severities themselves \citep{LiuWang2017,CossetteEtAl2019,OhAhnLee2021,blier-wong2024collective}. In the copula-linked compound regressions of \citet{shi2020regression}, the authors allow for dependence between the count and the severities but treat the individual claims as conditionally independent given the count, whereas claims within a policy are typically positively dependent. Parameter estimation in this second class has remained open, because each policy's likelihood involves a copula whose dimension is the random claim count. The second strand is size-biased and weighted sampling. The pooled-claims limit law is a special case of weighted distributions \citep{Rao1965,PatilRao1978}, for which \citet{Vardi1982,Vardi1985} developed nonparametric maximum likelihood, with asymptotic theory in \citet{GillVardiWellner1988}. See also \citet{ArratiaGoldsteinKochman2019} for a survey of size-biased transforms. In actuarial science, size-biased transforms arise prominently in conditional mean risk sharing \citep{Denuit2019}. The third strand is composite likelihood estimation in insurance. \citet{ShiFengBoucher2016} and \citet{YangShi2019} use pairwise likelihood in copula-based insurance models to avoid high-dimensional integration, and \citet{CossetteGadouryMarceauRobert2019} apply it to hierarchical Archimedean copulas. To the best of our knowledge, composite likelihood has not previously been used to estimate a CRM, and the size-biased structure of the pooled claims has not been recognized as the relevant sampling mechanism.

The remainder of this paper is structured as follows. Section~\ref{sec:background} sets up the dependent CRM and proves the convergence of the pooled claims to the size-biased law $F_Y$. Section~\ref{sec:cl} develops the composite likelihood and its policy-level Godambe information, together with the stepwise extension and hypothesis tests. Section~\ref{sec:sarmanovmodel} constructs the Sarmanov CRM and derives its closed-form observed-claim law. Section~\ref{sec:sarmanovest} specializes the estimation theory to the Sarmanov CRM and discusses implementation. Section~\ref{sec:numericalstudy} presents the simulation study, and Section~\ref{sec:conclusion} provides the conclusion. 

\section{Pooled claims sample a size-biased law}\label{sec:background}

This section establishes the problem that the rest of the paper solves. We define the class of dependent CRMs, explain why their full likelihood is intractable, prove that the pooled claims sample the size-biased law $F_Y$, and derive the resulting inconsistency of margin-first estimation.

\subsection{Collective risk models with dependent components}\label{subsec:model}

Let $N$ have cdf $\FN(n)=\p(N\le n)$ and pmf $\pN(n)=\p(N=n)$ on $\Nzero$, and let $X$ have continuous cdf $\FX$ on $\R_+$ with density $\fX$; extensions to discrete or mixed severities require only minor modifications. Following \citet{CossetteEtAl2019,OhAhnLee2021}, a CRM with dependence between the count and the individual severities specifies the joint distribution of $(N,X_1,\dots,X_k)$ for every $k\ge 1$. We express it through a sequence of copulas: for each $k\ge 1$,
\begin{equation}\label{eq:sklarMixed}
\p(N\le n, X_1\le x_1,\dots,X_k\le x_k) = C_{k+1}\left(\FN(n),\FX(x_1),\dots,\FX(x_k)\right),
\end{equation}
for a copula $C_{k+1}$ on $[0,1]^{k+1}$, which may depend on a finite-dimensional parameter vector to be estimated. Sklar's theorem \citep{Sklar1959} guarantees the existence of such a copula for each $k\ge 1$. Because $N$ is discrete, $C_{k+1}$ is not unique off the range of the margins; this is the usual subcopula phenomenon, see \citet[Chapter~2]{Nelsen2006} and \citet{GenestNeslehova2007}. The joint probabilities determined on the support are nonetheless well defined. The family \eqref{eq:sklarMixed} must also be internally consistent, so that the copulas $(C_{k+1})_{k\ge 1}$ describe a single law for the whole sequence $(N,X_1,X_2,\dots)$. For every $k\ge 1$, dropping the last severity from $(N,X_1,\dots,X_{k+1})$ must return the distribution of $(N,X_1,\dots,X_k)$. Since $\FX(x_{k+1})\to1$ as $x_{k+1}\to\infty$, this marginalization sets the last severity argument of $C_{k+2}$ to one, so consistency requires $C_{k+2}(u_0,u_1,\dots,u_k,1)=C_{k+1}(u_0,u_1,\dots,u_k)$ for every $k\ge 1$. We impose this only at the arguments the margins actually attain, since $N$ is discrete and the copula is identified nowhere else. We treat this Kolmogorov consistency as part of the model specification throughout; in the Sarmanov class, this consistency constrains the model parameters.

A policy-year observation is the variable-length vector $Z=(N, X_1,\dots,X_N)$, and the portfolio supplies i.i.d.\ copies $Z_1,\dots,Z_m$. The observation scheme is the central feature of the data. A policy reveals only its first $N$ severities, so a policy with $N=0$ contributes no severities and a policy with a large count contributes many. Under independence, this is harmless, because every observed severity has distribution $\FX$ regardless of the count that produced it. Under dependence, the conditional law of a severity varies with $N$, and the pool of observed claims mixes those conditional laws in proportion to the counts.

A model of the form \eqref{eq:sklarMixed} can be parameter-rich. In the FGM class, the $(k+1)$-variate copula has $2^{k+1}-(k+1)-1$ interaction parameters, one for each subset of coordinates of size at least two, and Sarmanov-type expansions behave similarly. An unstructured specification thus has exponentially many parameters in the maximal claim count. We therefore impose the following symmetry: within a policy, the order in which the claims are indexed carries no information about their sizes.

\begin{assumption}[Conditional exchangeability]\label{ass:exchangeability}
For each $n\ge 1$, conditional on $N=n$, the vector $(X_1,\dots,X_n)$ is exchangeable:
$$(X_1,\dots,X_n)\mid (N=n)\stackrel{d}{=}(X_{\sigma(1)},\dots,X_{\sigma(n)})\mid (N=n)$$
for every permutation $\sigma$ of $\{1,\dots,n\}$. The copulas $(C_{k+1})_{k\ge 1}$ are exchangeable in the severity coordinates, so the dependence parameters linking $N$ to an individual severity do not depend on the claim index.
\end{assumption}

Exchangeability of the copulas $C_{k+1}$ in their severity coordinates implies that the severity vector $(X_1,\dots,X_n)$ is exchangeable given $N=n$; we state both forms of Assumption~\ref{ass:exchangeability} because later arguments use each one. Assumption~\ref{ass:exchangeability} excludes mechanisms in which claim order is informative, such as first-claim scrutiny, deductible erosion, or reporting-order effects; those settings require order-specific margins in place of the exchangeable reduction. Under exchangeability, the dependence structure reduces to a small collection of parameters indexed by subset size: a single frequency--severity parameter linking $N$ to any severity, a single severity--severity parameter, a three-way parameter linking $N$ to a pair of severities, and so on. The low-order parameters govern the quantities that matter more in practice. The parameters entering the law of an observed claim, and hence $\E[S]$, are low-order, whereas parameters such as $\theta_{123}$ affect the variance and tail but are harder to estimate from typical portfolios. Whatever the ultimate target, an accurate estimate of the marginal behaviour under the size-biased observation scheme is a prerequisite.

\subsection{Intractability of the full likelihood}\label{subsec:fulllik}

The statistically natural estimator maximizes the full likelihood of the policy vectors $Z_i=(N_i,X_{i1},\dots,X_{iN_i})$, fitting the frequency, severity, and dependence parameters jointly. Were it tractable, it would be the preferred approach. However, as seen from the contribution of a single policy under \eqref{eq:sklarMixed}, this approach will be intractable. We fix $k\ge 1$, write $v_n=\FN(n)$ and $v_{n^-}=\FN(n-1)$ (with $v_{-1}:=0$), let $u_j=\FX(x_j)$, and define the mixed derivative operator
$$D_{1:k}C_{k+1}(u_0,u_1,\dots,u_k) := \frac{\partial^k}{\partial u_1\cdots \partial u_k} C_{k+1}(u_0,u_1,\dots,u_k).$$
Assume that for each $n$ with $\pN(n)>0$ the conditional law of $(X_1,\dots,X_k)$ given $N=n$ is absolutely continuous with respect to Lebesgue measure on $\R_+^k$; then the finite difference $C_{k+1}(v_n,\cdot)-C_{k+1}(v_{n^-},\cdot)$ is the integral of its almost-everywhere mixed derivative, and the joint pmf--pdf of $(N,X_1,\dots,X_k)$ is
\begin{equation}\label{eq:mixedlik}
f_{N,X_1,\dots,X_k}(n,x_1,\dots,x_k) = \left[D_{1:k}C_{k+1}(v_n,u_1,\dots,u_k)-D_{1:k}C_{k+1}(v_{n^-},u_1,\dots,u_k)\right]\prod_{j=1}^k \fX(x_j).
\end{equation}
The absolute-continuity requirement holds automatically for the polynomial expansions of Sections~\ref{sec:sarmanovmodel} and \ref{sec:sarmanovest} and must be checked family by family beyond them. The finite difference in the first argument distinguishes \eqref{eq:mixedlik} from the familiar continuous copula density; it is unavoidable when $N$ is discrete.

The full likelihood of a policy with $n$ claims evaluates \eqref{eq:mixedlik} at $k=n$. It therefore requires an $n$th-order mixed derivative of an $(n+1)$-dimensional copula, one such term for every observed count, and in subset-expansion families the number of interaction parameters grows exponentially with $n$. Unless the counts stay small and the high-dimensional copula derivatives are available in closed form, direct maximum likelihood is impractical. Existing estimation methods, such as the method of moments, IFM, and composite likelihood variants tailored to specific families or to low claim counts, avoid this computation by fitting $\FX$ to the pooled severities. We now show that this shortcut estimates the wrong distribution, by identifying the law the pooled claims actually sample.

Let $Z_1,\dots,Z_m$ be i.i.d.\ policy observations, $Z_i=(N_i,X_{i1},\dots,X_{iN_i})$, where $i=1,\dots,m$ indexes policies and $j=1,\dots,N_i$ indexes the claims of policy $i$, with $Z_i=(N_i)$ carrying no severities when $N_i=0$; this indexing convention is used throughout. Let
$$N_\bullet = \sum_{i=1}^m N_i$$
be the total number of observed claims. The pooled empirical measure and cdf of the claim severities are
$$\widehat\mu_m(A) = \frac{1}{N_\bullet}\sum_{i=1}^m\sum_{j=1}^{N_i}\ind\{X_{ij}\in A\}, \qquad A\in\mathcal{B}(\R_+),$$
and
$$\widehat{F}_{m}(x) = \frac{1}{N_\bullet}\sum_{i=1}^m\sum_{j=1}^{N_i}\ind\{X_{ij}\le x\}, \quad x\in\R_+,$$
where $\mathcal{B}(\R_+)$ denotes the Borel $\sigma$-algebra on $\R_+$, with an arbitrary convention when $N_\bullet=0$.

The limiting law of the pool is the probability measure of an arbitrary observed claim,
\begin{equation}\label{eq:FYmeasure}
\mu_Y(A) := \frac{\E\left[\sum_{j=1}^{N}\ind\{X_j\in A\}\right]}{\E[N]}, \qquad A\in\mathcal{B}(\R_+),
\end{equation}
with cdf $F_Y(x)=\mu_Y([0,x])$. The name reflects the following sampling experiment: select one claim uniformly at random among all claims of a large portfolio; $\mu_Y$ is its law. A policy with $n$ claims is selected with probability proportional to $n$, so high-frequency policies are over-represented, and claim-free policies do not appear. The next theorem shows that the pooled claims converge to $\mu_Y$, thereby identifying the law of the pooled claims sample and, with it, the data-generating mechanism on which the estimation procedure in Section~\ref{sec:cl} is built.

\begin{theorem}\label{thm:FYlimit}
Assume $\E[N]\in(0,\infty)$. Then $\widehat\mu_m\Rightarrow \mu_Y$ almost surely; in particular, for every continuity point $x$ of $F_Y$,
$$\widehat{F}_{m}(x)\longrightarrow F_Y(x),\quad\text{a.s.}$$
Under Assumption~\ref{ass:exchangeability},
\begin{equation}\label{eq:FYmixture}
F_Y(x) = \sum_{n=1}^\infty \frac{n\pN(n)}{\E[N]}F_{X\mid N=n}(x),
\end{equation}
where, for $n\ge 1$, $F_{X\mid N=n}$ is the common conditional cdf of $X_1,\dots,X_n$ given $N=n$.
\end{theorem}

\begin{proof}
For a fixed Borel set $A$, write $\widehat\mu_m(A)=A_m(A)/B_m$, where
$$A_m(A)=\sum_{i=1}^m\sum_{j=1}^{N_i}\ind\{X_{ij}\in A\}, \qquad B_m=\sum_{i=1}^m N_i.$$
Both are sums of i.i.d.\ policy-level contributions with finite expectations, so the strong law of large numbers gives $A_m(A)/m\to \E[\sum_{j=1}^{N}\ind\{X_j\in A\}]$ and $B_m/m\to \E[N]$ a.s. Applying this to the sets $A=[0,x]$ for $x$ in a countable dense set $D\subset\R_+$ of continuity points of $F_Y$ gives, on a single probability-one event, $\widehat{F}_m(x)\to F_Y(x)$ for all $x\in D$; monotonicity of $\widehat{F}_m$ and $F_Y$ extends the convergence to every continuity point of $F_Y$, which is equivalent to $\widehat\mu_m\Rightarrow \mu_Y$ on that event. The mixture form \eqref{eq:FYmixture} follows by conditioning on $N=n$ and applying exchangeability. Finally, each conditional cdf $F_{X\mid N=n}(x)=\{C_2(\FN(n),\FX(x))-C_2(\FN(n-1),\FX(x))\}/\pN(n)$ is continuous in $x$, since $\FX$ is continuous and copulas are Lipschitz in each coordinate; hence $F_Y$ is continuous, its set of continuity points is all of $\R_+$, and the convergence holds at every $x\ge 0$.
\end{proof}

The identity \eqref{eq:FYmixture} shows exactly when pooling is harmless. In the independent CRM, $F_{X\mid N=n}=\FX$ for all $n\ge 1$, the mixture collapses, and $F_Y=\FX$: the classical workflow is consistent. Under dependence, both the weights $n\pN(n)/\E[N]$, $n\ge 1$, and the components $F_{X\mid N=n}$, $n\ge 1$, pull $F_Y$ away from $\FX$. The limit is a special case of length-biased or weighted sampling; see \citet{Rao1965,PatilRao1978,Vardi1982,Vardi1985,GillVardiWellner1988,ArratiaGoldsteinKochman2019}.

A consequence used repeatedly below expresses claim-level sums as integrals against $F_Y$.

\begin{corollary}\label{cor:FYgeneral}
Assume the conditions of Theorem~\ref{thm:FYlimit} and let $h:\R_+\to\R$ be measurable with $\E\left[\sum_{j=1}^N |h(X_j)|\right]<\infty$. Then
$$\E\left[\sum_{j=1}^N h(X_j)\right] = \E[N]\int_0^\infty h(x)\d F_Y(x).$$
Under Assumption~\ref{ass:exchangeability}, the integrability condition reduces to $\E[N\,|h(X_1)|]<\infty$, and
$$\E\left[\sum_{j=1}^N h(X_j)\right] = \sum_{n=1}^\infty n\pN(n)\int_0^\infty h(x)\d F_{X\mid N=n}(x).$$
\end{corollary}

\begin{proof}
The first identity is the integral form of the definition \eqref{eq:FYmeasure} of $\mu_Y$: it holds for indicator functions $h=\ind_A$, extends to nonnegative measurable $h$ by monotone convergence, and to integrable signed $h$ by decomposition into positive and negative parts. Under conditional exchangeability, $\E[\sum_{j=1}^N |h(X_j)|]=\sum_{n\ge 1}n\pN(n)\E[|h(X_1)|\mid N=n]=\E[N\,|h(X_1)|]$, which gives the reduced integrability condition; the same conditioning with $h$ in place of $|h|$ yields the mixture identity $\E[\sum_{j=1}^N h(X_j)]=\sum_{n\ge 1}n\pN(n)\int_0^\infty h(x)\d F_{X\mid N=n}(x)$.
\end{proof}

\subsection{The inconsistency of naive pooled fitting}\label{subsec:naive}

By Theorem~\ref{thm:FYlimit}, any procedure that fits a severity model to the pooled claims estimates $F_Y$, not $\FX$. The next corollary makes this concrete when the severity margin is fitted parametrically.

\begin{corollary}\label{cor:naiveMLE}
Consider the margin-first estimator that fits a parametric density $\fX(\cdot;\psi)$ to the pooled claims by maximizing
$$L_m^{\mathrm{naive}}(\psi)=\prod_{i=1}^m\prod_{j=1}^{N_i}\fX(X_{ij};\psi),$$
treating the pooled severities as a sample from $\FX(\cdot;\psi)$. Under standard M-estimation regularity conditions, in the sense of \citet[Chapters~5 and 19]{vanDerVaart1998}, any maximizer converges almost surely to
$$\psi^\star = \operatorname*{arg\,max}_{\psi\in\Psi}\int \log \fX(x;\psi)\,\d F_Y(x),$$
the Kullback--Leibler projection of $F_Y$ onto the fitted severity family. The estimator is consistent for the true value $\psi_0$ exactly when $\psi^\star=\psi_0$. When $F_Y\ne\FX(\cdot;\psi_0)$, the projection generically differs from $\psi_0$, so the margin-first estimator is inconsistent for the marginal severity law.
\end{corollary}

\begin{proof}
Applying Corollary~\ref{cor:FYgeneral} to $h(x)=\log \fX(x;\psi)$, the criterion normalized by $m$ converges almost surely, uniformly over the compact $\Psi$ under a standard envelope condition \citep[Chapter~19]{vanDerVaart1998}, to
$$\E[N]\int\log\fX(x;\psi)\d F_Y(x).$$
The argmax theorem for misspecified M-estimators \citep[Chapter~5]{vanDerVaart1998} then yields convergence of any maximizer to the unique maximizer $\psi^\star$ of the limit criterion.
\end{proof}

Corollary~\ref{cor:naiveMLE} covers every procedure whose first stage fits $\FX$ to uncorrected pooled claims. Standard two-stage copula estimators, such as the IFM method of \citet{Joe2005} and the semiparametric pseudo-likelihood and rank-based procedures of \citet{GenestGhoudiRivest1995} and \citet{Tsukahara2005}, were developed for settings with one observation per unit and correctly targeted margins. They are consistent there. Applied to pooled CRM claims, however, their first stage estimates $F_Y$ rather than $\FX$, with an error that does not shrink with the sample size; the size-bias correction we develop below is needed before such methods can be used. The corollary does not apply to procedures that model the observation mechanism correctly, use a conditional severity law, or sample one claim per policy. Copula inference with non-continuous margins raises further identifiability issues; see \citet{NasriRemillard2023}.

The full policy likelihood is therefore impractical, and fitting the severity margin to the pooled claims is inconsistent under dependence. However, we can use the distribution of the pooled claims to estimate the parameters. By Theorem~\ref{thm:FYlimit} this is $F_Y$, a low-dimensional, model-determined law. Whenever the $F_Y$ implied by the model has a tractable form, the pooled claims provide a valid likelihood for exactly the parameters that enter $F_Y$. We build that likelihood in Section~\ref{sec:cl}.

\section{Composite likelihood built on the observed-claim law}\label{sec:cl}

In this section, we develop a composite likelihood estimator that is consistent for the marginal severity law and the low-order dependence parameters. The estimator is based on the two low-dimensional components of the CRM: the count likelihood and the observed-claim likelihood. The former identifies the frequency parameters, and the latter identifies the severity and low-order dependence parameters. The composite likelihood is a sum of these two components, and its maximizer is consistent for the true parameter vector under standard regularity conditions. The asymptotic variance is a Godambe sandwich, which we estimate at the policy level to account for within-policy dependence. A stepwise extension allows estimation of higher-order dependence parameters from conditional margins, with a stacked sandwich that propagates uncertainty from earlier steps.

\subsection{From the full likelihood to low-dimensional components}\label{subsec:clsetup}

Let $\fullparam$ denote the full parameter of a copula-based CRM: frequency parameters for $\FN$, severity parameters for $\FX$, and the dependence parameters of $C_{k+1}$ for every $k\ge 1$. Direct likelihood inference for $\fullparam$ would require the full policy likelihood of Section~\ref{subsec:fulllik}. The CRM instead offers two components that are both directly available from the data and low-dimensional. The claim counts have marginal law $\pN$, undistorted by the observation scheme, and the pooled claims have law $F_Y$ by Theorem~\ref{thm:FYlimit}. When $F_Y$ depends only on a low-dimensional parameter subset, these two components identify that subset on their own, so it can be estimated without specifying or fitting the high-order dependence parameters.

We collect the targeted parameters $\bm{\phi}=(\eta,\psi,\delta) \subset \fullparam$, where $\eta$ parameterizes $\FN(\cdot;\eta)$, $\psi$ parameterizes $\FX(\cdot;\psi)$, and $\delta$ denotes the low-order dependence parameters entering the law of $Y$. Higher-order dependence parameters are handled by the stepwise scheme of Section~\ref{subsec:stepwise}. Throughout, the parameterization is variation-free: $\pN(\cdot;\eta)$ depends only on $\eta$ and $\FX(\cdot;\psi)$ only on $\psi$, so the marginal parameters and the dependence parameters range over a product space: any admissible value of one imposes no restriction on the others, as is standard in copula models.

The model-implied observed-claim law follows from \eqref{eq:FYmixture}:
$$F_Y(x;\bm{\phi}) = \sum_{n=1}^\infty \frac{n\pN(n;\eta)}{\mu_N(\eta)}F_{X\mid N=n}(x;\psi,\delta), \qquad \mu_N(\eta)=\E_\eta[N].$$
Whenever the conditional laws $F_{X\mid N=n}$, $n\ge 1$, are tractable, so is $F_Y$. In many families, including the exchangeable FGM and Sarmanov constructions, $F_{X\mid N=n}$ depends only on the bivariate dependence between $N$ and a single claim, so $F_Y$ involves no dependence parameter beyond $\delta$. We formalize the requirement.

\begin{assumption}[Tractable observed-claim model]\label{ass:FYtractable}
There is a parametric family $\{F_Y(\cdot;\bm{\phi})\}$ such that, under the CRM with parameter $\bm{\phi}_0=(\eta_0,\psi_0,\delta_0)$, the arbitrary observed claim $Y$ has cdf $F_Y(\cdot;\bm{\phi}_0)$. For each $\bm{\phi}$, $F_Y(\cdot;\bm{\phi})$ is absolutely continuous with density $f_Y(\cdot;\bm{\phi})$, and $\bm{\phi}\mapsto f_Y(x;\bm{\phi})$ is differentiable for almost every $x\ge 0$.
\end{assumption}

In the Sarmanov model, $f_Y(\cdot;\bm{\phi})$ is a density on an admissible parameter set, made explicit in Section~\ref{sec:sarmanovmodel}. Families without closed-form conditional laws still admit a numerically evaluated $f_Y$, and the asymptotic theory below applies as long as $f_Y$ is evaluable and smooth.

\subsection{The estimator}\label{subsec:clest}

Replacing the full policy density by the count marginal $\pN$ and the observed-claim law $f_Y$ gives a composite log-likelihood in the counts $N_1,\dots,N_m$ and pooled severities $\{X_{ij}:1\le i\le m,\ 1\le j\le N_i\}$,
\begin{equation}\label{eq:clloglik}
\ell_{\CL}(\bm{\phi}) = \sum_{i=1}^m \log \pN(N_i;\eta) + \sum_{i=1}^m\sum_{j=1}^{N_i}\log f_Y(X_{ij};\bm{\phi}),
\end{equation}
and let $\widehat{\bm{\phi}}_{\CL}$ be a maximizer. The second sum treats each observed claim as an independent draw from $f_Y(\cdot;\bm{\phi})$. It is therefore not the true joint log-density of a policy's claims, since claims within a policy are dependent and their number is random, but the construction does not require it to be. What matters is that each pooled claim has marginal law $F_Y$ by Theorem~\ref{thm:FYlimit}, so \eqref{eq:clloglik} is a valid composite-likelihood criterion: an M-estimation objective over i.i.d.\ policy clusters whose population version is maximized at the true parameter when $f_Y$ is correctly specified; this is the defining logic of composite likelihood \citep{Lindsay1988,CoxReid2004,Varin2008,VarinReidFirth2011,Bhat2014}. Ignoring the within-cluster dependence costs efficiency, not consistency. It does invalidate the usual variance formula: because the composite score aggregates correlated components, the information matrix equality between the expected Hessian and the score variance fails, and the inverse Hessian of $\ell_{\CL}$ does not estimate the variance. The correct variance is the Godambe sandwich below.

\subsection{Asymptotics and the policy-level Godambe information}\label{subsec:godambe}

We write the contribution of a single policy (a cluster) to the composite log-likelihood \eqref{eq:clloglik}, together with its score, as
$$q(Z_i;\bm{\phi}) = \log\pN(N_i;\eta)+\sum_{j=1}^{N_i}\log f_Y(X_{ij};\bm{\phi}),
\qquad U_i(\bm{\phi})=\nabla_{\bm{\phi}} q(Z_i;\bm{\phi}), \qquad i=1,\dots,m. $$
and let $\ell_{\CL,i}(\bm\phi)=q(Z_i;\bm\phi)$. The parameter $\bm{\phi}$ ranges over a compact parameter space $\Phi$, with interior $\mathrm{int}(\Phi)$. We impose standard compactness, smoothness, envelope, identification, and nonsingularity assumptions of the type used for M-estimators \citep[Chapter~5]{vanDerVaart1998}. Because each cluster contains a random number of claim terms, the envelopes are imposed at the policy-cluster level, and the attendant uniform laws of large numbers are understood in the empirical-process sense of \citet[Chapter~19]{vanDerVaart1998}; consistency requires only $\E[N]<\infty$ while asymptotic normality requires $\E[N^2]<\infty$. We refer to these assumptions as the standard M-estimation regularity conditions.

\begin{theorem}\label{thm:CLasymp}
Suppose the policy observations are i.i.d.\ from the CRM with $\bm{\phi}_0\in\mathrm{int}(\Phi)$, and that $f_Y(\cdot;\bm{\phi})$ is correctly specified, identifiable, and regular under the standard M-estimation regularity conditions. Then any maximizer $\widehat{\bm{\phi}}_{\CL}$ of \eqref{eq:clloglik} satisfies $\widehat{\bm{\phi}}_{\CL}\plim\bm{\phi}_0$ and
$$\sqrt{m}(\widehat{\bm{\phi}}_{\CL}-\bm{\phi}_0) \dto \mathcal{N}(0,\Sigma_{\CL}),
\qquad \Sigma_{\CL} = \Hmat(\bm{\phi}_0)^{-1}\Jmat(\bm{\phi}_0)\Hmat(\bm{\phi}_0)^{-1},$$
where
$$ \Hmat(\bm{\phi}) = -\E\left[\nabla^2_{\bm{\phi}}\ell_{\CL,i}(\bm{\phi})\right],
\qquad \Jmat(\bm{\phi}) = \E\left[U_i(\bm{\phi})U_i(\bm{\phi})^\top\right]. $$
\end{theorem}

\begin{proof}
The composite log-likelihood $\ell_{\CL}(\bm{\phi})=\sum_{i=1}^m q(Z_i;\bm{\phi})$ is an M-estimation criterion over i.i.d.\ policy clusters. The standard M-estimation regularity conditions give uniform convergence of $m^{-1}\ell_{\CL}$ to $Q(\bm{\phi})=\E[q(Z;\bm{\phi})]$ over the compact $\Phi$ and identify $\bm{\phi}_0$ as the unique maximizer; consistency follows from the argmax theorem, see \citet[Chapter~5]{vanDerVaart1998} and \citet[Theorem~2.1]{NeweyMcFadden1994}. Asymptotic normality with the Godambe covariance follows by Taylor expansion of the first-order condition around $\bm{\phi}_0$ and the multivariate central limit theorem applied to the i.i.d.\ policy scores; see \citet[Section~3.2]{VarinReidFirth2011}. Because the summand $q(Z;\bm{\phi})$ contains a random number of claim terms, the envelopes are stated at the cluster level, with the supremum inside the cluster sum, as required for the relevant uniform law of large numbers \citep[Chapter~19]{vanDerVaart1998}; and the cluster unit is the policy, so $\Jmat$ is computed from scores that aggregate all claims within a policy.
\end{proof}

Identifiability in Theorem~\ref{thm:CLasymp} is a model-specific requirement. When the cluster-level integrability condition of Corollary~\ref{cor:FYgeneral} holds for $h(x)=\log f_Y(x;\bm{\phi})$, the claim term of the population criterion satisfies
$\E\left[\sum_{j=1}^N \log f_Y(X_j;\bm{\phi})\right] = \E[N]\int \log f_Y(x;\bm{\phi})\d F_Y(x;\bm{\phi}_0)$,
so it is, up to the factor $\E[N]$, the expected log-likelihood of one draw from the true $F_Y$. Identification then reduces to a Kullback--Leibler argument within the family $\{f_Y(\cdot;\bm{\phi})\}$. Theorem~\ref{thm:sarmCL} verifies this identification for the Sarmanov model under standard M-estimation regularity conditions.

The variance in Theorem~\ref{thm:CLasymp} requires care in implementation, because two separate effects make $\Hmat^{-1}$ alone invalid. First, composite likelihood violates the information identity $\Hmat=\Jmat$ even under correct specification of each component, because the composite score aggregates correlated pieces of information. Second, the CRM's variable-length observation scheme adds dependence between the count term and the pooled-claim terms of the same policy. The sandwich captures both effects, provided the scores are aggregated at the policy level: $U_i$ sums the score contributions of all claims of policy $i$ before the outer product is taken. Replacing $U_iU_i^\top$ by a sum of per-claim outer products treats dependent claims as independent, understates $\Jmat$, and yields standard errors that are too small whenever within-policy severities are dependent. Section~\ref{subsec:sim-coherent} quantifies the resulting undercoverage.

Estimation of $\Sigma_{\CL}$ is standard: with
$$\widehat{\Hmat} = -\frac{1}{m}\sum_{i=1}^m \nabla^2_{\bm{\phi}} q(Z_i;\widehat{\bm{\phi}}_{\CL}), \qquad \widehat{\Jmat} = \frac{1}{m}\sum_{i=1}^m U_i(\widehat{\bm{\phi}}_{\CL})U_i(\widehat{\bm{\phi}}_{\CL})^\top,$$
the estimator $\widehat{\Sigma}_{\CL} = \widehat{\Hmat}^{-1}\widehat{\Jmat}\widehat{\Hmat}^{-1}$ is consistent under these regularity conditions, by the uniform law of large numbers and the continuous mapping theorem.

\subsection{Stepwise estimation of higher-order dependence}\label{subsec:stepwise}

The composite likelihood \eqref{eq:clloglik} estimates $(\eta,\psi,\delta)$, the parameters visible in the observed-claim law. When the analysis also requires within-policy dependence parameters, for instance for the variance of $S$ or tail risk measures, we estimate them stepwise from low-dimensional conditional margins. The procedure exploits a structural property of copulas built as subset expansions, such as the FGM and Sarmanov families of Section~\ref{sec:sarmanovmodel}: their low-order margins depend only on the interaction parameters of the retained coordinates.

\begin{assumption}[Tractable low-order margins]\label{ass:loworder}
There exist known functions $h_1$ and $h_2$ such that, under the CRM:
(i) $(N,X_1)$ has pmf--pdf $h_1(n,x;\eta,\psi,\theta_{01})$ for $n\ge 1$ and $x>0$;
(ii) $(N,X_1,X_2)$ has pmf--pdf $h_2(n,x_1,x_2;\eta,\psi,\theta_{01},\theta_{12},\theta_{012})$ for $n\ge 2$ and $x_1,x_2>0$;
(iii) these margins involve no dependence parameters beyond those appearing as arguments of $h_1$ and $h_2$ in (i) and (ii).
\end{assumption}

Assumption~\ref{ass:loworder} holds in the exchangeable FGM and Sarmanov constructions of Section~\ref{sec:sarmanovmodel}, where integrating out a coordinate removes every interaction term that involves it.

We estimate the parameters in three steps, ordered from the lowest-dimensional margin to the highest and mirroring the way actuaries build a model: first the frequency, then the severity together with its dependence on the count, and finally the within-policy severity dependence. Each step reuses the estimates from earlier steps.

In the first step, we fit the frequency model. The marginal law of $N$ involves no severity or dependence parameter, so
$$\widehat{\eta} \in \operatorname*{arg\,max}_{\eta}\sum_{i=1}^m \log \pN(N_i;\eta)$$
is an ordinary MLE, consistent and asymptotically normal under standard count-model regularity \citep[Chapter~5]{vanDerVaart1998}.

In the second step, we estimate the severity margin and the frequency--severity dependence. With $\widehat\eta$ plugged in, we maximize the bivariate composite objective
$$\ell_2(\psi,\theta_{01};\widehat{\eta}) = \sum_{i=1}^m\sum_{j=1}^{N_i}\log h_1(N_i,X_{ij};\widehat{\eta},\psi,\theta_{01}) - \sum_{i=1}^m N_i\log \pN(N_i;\widehat{\eta})$$
over $(\psi,\theta_{01})$. The subtracted term is constant in $(\psi,\theta_{01})$ and only normalizes the objective into a sum of conditional log-densities $\log f_{X\mid N}(X_{ij}\mid N_i)$, which is convenient for composite-likelihood information criteria \citep{VarinVidoni2005} and likelihood-ratio statistics. We note, however, that maximizing the first sum alone gives the same estimating equations. Using every claim in Step~2 exploits all valid instances with at least one claim. This is not the same sampling scheme as taking only the first claim of each policy with $N_i\ge 1$: the pooled claims are weighted by $n\pN(n)$ and have law $F_Y$. It would be possible to develop an estimation procedure using only the first claim; in that case, first claims are weighted by $\pN(n)$ conditional on $N\ge 1$, a different law under frequency--severity dependence. The first-claim design would need its own sampling correction and would discard the remaining claims. Using all of them through the conditional margin $f_{X\mid N}$ gains efficiency, and the induced within-policy dependence in the score is absorbed by the policy-level sandwich.

Steps~1 and~2 are already enough to compute the mean of $S$. Whenever $\E[N]\in(0,\infty)$ and $\E[S]<\infty$, Corollary~\ref{cor:FYgeneral} with $h(x)=x$ gives
$$\E[S]=\E[N]\,\E[Y], \qquad \E[Y]=\int_0^\infty x\,f_Y(x;\bm{\phi})\,\d x,$$
so the fitted count model supplies $\E[N]$ and the fitted observed-claim law supplies $\E[Y]$, with no higher-order dependence parameter involved; in the Sarmanov model both admit closed forms (Proposition~\ref{prop:sarm-mean}).

In the third step, we estimate the within-policy dependence. For $n\ge 2$, let
$$k(x_2\mid n,x_1;\eta,\psi,\theta_{01},\theta_{12},\theta_{012}) =
\frac{h_2(n,x_1,x_2;\eta,\psi,\theta_{01},\theta_{12},\theta_{012})}{h_1(n,x_1;\eta,\psi,\theta_{01})},$$
and, with the first two steps plugged in, we maximize
\begin{equation}\label{eq:step3}
\ell_3(\theta_{12},\theta_{012};\widehat{\eta},\widehat{\psi},\widehat{\theta}_{01}) = \sum_{i: N_i\ge 2}\log k(X_{i2}\mid N_i,X_{i1};\widehat{\eta},\widehat{\psi},\widehat{\theta}_{01},\theta_{12},\theta_{012}).
\end{equation}
In Step~3, we chose not to sample each possible pair of claims, but to focus on the first pair. The target is the dependence parameters that first appear in the trivariate margin $h_2$, and by exchangeability any within-policy pair has the same low-order law, so the first two claims give one valid contribution from each policy with $N_i\ge 2$. Reusing all $\binom{N_i}{2}$ pairs is also possible, but it weights high-count policies quadratically and introduces strongly dependent pair terms. Under correct specification, this changes the composite-likelihood weighting, but often adds little information. We therefore use one pair per eligible policy as the default. Either way, the sandwich variance is valid only when computed at the policy level.

Inference for the stepwise estimator must account for the plug-in uncertainty from earlier steps. We group the parameters by step, $\bm{\theta}^{(1)}=\eta$, $\bm{\theta}^{(2)}=(\psi,\theta_{01})$, $\bm{\theta}^{(3)}=(\theta_{12},\theta_{012})$, write $\bm{\theta}=(\bm{\theta}^{(1)},\bm{\theta}^{(2)},\bm{\theta}^{(3)})$, and stack the policy-level estimating functions of the three steps,
$$U_i(\bm{\theta})=
\begin{pmatrix}
U_i^{(1)}(\bm{\theta}^{(1)})\\
U_i^{(2)}(\bm{\theta}^{(2)};\bm{\theta}^{(1)})\\
U_i^{(3)}(\bm{\theta}^{(3)};\bm{\theta}^{(1)},\bm{\theta}^{(2)})
\end{pmatrix},
\qquad i=1,\dots,m.$$
so that the stepwise estimator solves $\sum_{i=1}^m U_i(\widehat{\bm{\theta}})=0$.

\begin{theorem}\label{thm:stepwise}
Under Assumptions~\ref{ass:exchangeability} and~\ref{ass:loworder}, suppose the data are generated at the true parameter $\bm{\theta}_0$, interior to the parameter space, and the stacked estimating functions are regular under standard M-estimation conditions \citep[Chapter~5]{vanDerVaart1998}. Then the stepwise estimator $\widehat{\bm{\theta}}=(\widehat{\bm{\theta}}^{(1)},\widehat{\bm{\theta}}^{(2)},\widehat{\bm{\theta}}^{(3)})$ satisfies $\widehat{\bm{\theta}}\plim\bm{\theta}_0$ and
$$\sqrt{m}(\widehat{\bm{\theta}}-\bm{\theta}_0) \dto \mathcal{N}(0,\Sigma_{\mathrm{SW}}),$$
with the multi-step sandwich covariance
\begin{equation}\label{eq:SWsandwich}
\Sigma_{\mathrm{SW}} = A^{-1} B (A^{-1})^\top,
\qquad
A=\E[\nabla_{\bm{\theta}}U_i(\bm{\theta}_0)],
\quad
B=\Var(U_i(\bm{\theta}_0)).
\end{equation}
\end{theorem}

\begin{proof}
The estimator solves the stacked equations $\sum_{i=1}^mU_i(\widehat{\bm{\theta}})=0$. We first verify that the stacked functions have mean zero at $\bm{\theta}_0$. The Step-1 component is the count score, with $\E[\nabla_\eta\log\pN(N;\eta_0)]=0$ by standard likelihood theory. For Step~2, $h_1(n,x)=\pN(n;\eta)f_{X\mid N=n}(x)$, so the $(\psi,\theta_{01})$-score of $\log h_1$ is the score of the conditional density $f_{X\mid N=n}$; by conditional exchangeability,
\begin{align*}
\E\left[U_i^{(2)}(\bm{\theta}^{(2)}_0;\bm{\theta}^{(1)}_0)\mid N_i=n\right]
&= n\int_0^\infty \nabla_{(\psi,\theta_{01})}\log f_{X\mid N=n}(x)\, f_{X\mid N=n}(x)\d x\\
&= n\,\nabla_{(\psi,\theta_{01})}\int_0^\infty f_{X\mid N=n}(x)\d x=0,
\end{align*}
with differentiation under the integral justified by standard smoothness and envelope assumptions for likelihood scores \citep[Chapter~5]{vanDerVaart1998}. For Step~3, $k(\cdot\mid n,x_1)$ is a density in $x_2$ for every $(n,x_1)$, so $\E[U_i^{(3)}\mid N_i,X_{i1}]=0$ on $\{N_i\ge 2\}$ by the same argument; taking expectations gives $\E[U_i(\bm{\theta}_0)]=0$. Consistency is sequential: Step~1 is an MLE, and Steps~2 and~3 inherit consistency through continuity of the plug-in mapping by multi-step extremum-estimator arguments \citep[Section~6]{NeweyMcFadden1994}. Asymptotic normality follows from the multivariate central limit theorem for $m^{-1/2}\sum_{i=1}^m U_i(\bm{\theta}_0)$ and a Taylor expansion of the stacked equations around $\bm{\theta}_0$; the block lower-triangular form of $A$ reflects the dependence of later scores on earlier plug-in estimates.
\end{proof}

The matrix $A$ is block lower-triangular because each step conditions on the estimates of the preceding steps. Treating Step~3 as a standalone likelihood with fixed nuisance parameters ignores this propagation and understates the variability of $(\widehat{\theta}_{12},\widehat{\theta}_{012})$; in implementation, every score contribution is a policy contribution.

Two composite-likelihood constructions now coexist, and their roles differ. The marginal composite likelihood \eqref{eq:clloglik}, built from $f_Y$, estimates $(\eta,\psi,\theta_{01})$ jointly and needs nothing beyond the observed-claim law. The stepwise construction uses conditional low-order margins, $h_1$ in Step~2 and $h_2/h_1$ in Step~3, and additionally estimates $(\theta_{12},\theta_{012})$. In the Sarmanov model, both consistently estimate $(\eta,\psi,\theta_{01})$ under their respective conditions.

\subsection{Tests for dependence parameters}\label{subsec:clrt}

Within the stepwise framework, we test scalar dependence parameters using Wald statistics based on the stacked policy-level sandwich estimator, which accounts for propagated plug-in uncertainty and avoids the nonstandard reference distributions of composite-likelihood ratio statistics \citep{ChandlerBate2007,PaceSalvanSartori2011}. For a scalar component $\gamma$ of $\bm{\theta}$ and null value $\gamma_0$,
\begin{equation}\label{eq:wald}
T_{\mathrm{Wald}} = \frac{(\widehat{\gamma}-\gamma_0)^2}{\widehat{\Sigma}_{\mathrm{SW},\gamma\gamma}/m}
\;\dto\; \chi^2_1 \quad\text{under } H_0,
\end{equation}
where $\widehat{\Sigma}_{\mathrm{SW},\gamma\gamma}$ is the corresponding diagonal entry of the estimated stacked sandwich. The test of no frequency--severity dependence, $H_0\colon\theta_{01}=0$, uses
$$T_{01} = \frac{m\,\widehat{\theta}_{01}^2}{\widehat{\Sigma}_{\mathrm{SW},\theta_{01}\theta_{01}}},$$
and the test of no three-way interaction, $H_0\colon\theta_{012}=0$, uses
$$T_{012} = \frac{m\,\widehat{\theta}_{012}^2}{\widehat{\Sigma}_{\mathrm{SW},\theta_{012}\theta_{012}}}.$$
In moderate samples, or when the estimate sits near the admissibility boundary, we recommend a parametric null bootstrap: fit the model under the null, simulate full policy clusters from the fitted null model, refit by the same stepwise procedure, and compare the observed statistic with the bootstrap distribution. A nonparametric bootstrap that resamples the observed policies would draw from the fitted alternative rather than the null, so its distribution is not a valid null reference for these tests.

Using the $\chi^2_1$ distribution as the null reference in \eqref{eq:wald} also requires the null value $\gamma_0$ to lie in the interior of the admissible set. For the Step-3 estimator the trivariate margin leaves $\theta_{12}$ and $\theta_{012}$ free over intervals containing $0$ in their interior, so the $\chi^2_1$ reference applies to the tests of no severity--severity dependence, $H_0\colon\theta_{12}=0$, and no three-way dependence, $H_0\colon\theta_{012}=0$. The reference fails only when an additional restriction places the null on a boundary. If one requires a projectively coherent model on unbounded count support, coherence forces $\theta_{12}\ge\theta_{01}^2$, so the null $\theta_{12}=0$ is feasible only at $\theta_{01}=0$, where it lies on the boundary. There one can use a parametric bootstrap or a one-sided boundary-corrected reference instead.

\section{The Sarmanov collective risk model}\label{sec:sarmanovmodel}

The composite likelihood needs a model in which $F_Y$ and, for the stepwise extension, the low-order margins $h_1$ and $h_2$ are explicit. The Sarmanov CRM provides all three in closed form. It generalizes the FGM CRM of \citet{blier-wong2024collective} while preserving the two properties the method uses: a subset expansion whose low-order margins involve only low-order parameters, and a latent Bernoulli-mixture representation that produces explicit conditional distributions and admissibility constraints. Sarmanov copulas have been widely used in actuarial science; see \citet{HashorvaRatovomirija2015} for mixed Erlang risks, \citet{AbdallahBoucherCossette2016} for dynamic claim counts, and \citet{BolanceVernic2019} for multivariate count regression.

\subsection{From FGM to Sarmanov}\label{subsec:fgm2sarm}

The FGM copula \citep{Eyraud1936,Morgenstern1956,Gumbel1960,Farlie1960} is the simplest family combining positive and negative dependence with an explicit multivariate subset expansion. In dimension two, its expression is given by
$$C_{\mathrm{FGM}}(u,v)=uv\left(1+\theta(1-u)(1-v)\right),\quad \theta\in[-1,1],$$
and in higher dimensions the family adds one interaction term per coordinate subset, and can be interpreted as a polynomial perturbation of the product copula \citep{Cambanis1977}. The family is also limited: the fixed kernel $u(1-u)$ bounds the attainable strength of dependence and excludes asymmetric shapes.

Sarmanov copulas \citep{Sarmanov1966,Lee1996} remove this restriction by allowing general kernels. We call a bivariate copula Sarmanov if its copula has the form
$$C(u_1,u_2)=u_1u_2 + a\,g_1(u_1)g_2(u_2)$$
for kernel functions $g_1,g_2$ with $g_i(0)=g_i(1)=0$ and a coefficient $a$ chosen so that $C$ is a valid copula. In higher dimensions, one may extend the family via a subset expansion analogous to the multivariate FGM, with general kernel factors. We adopt the formulation of \citet{blier-wong2026stochastic}, which represents these copulas through a stochastic representation based on a latent Bernoulli random vector. The representation supplies both the conditional distributions needed for estimation and explicit admissibility conditions. Appendix~\ref{app:bernoulli-sarmanov} gives the construction; here we record the two facts we use to characterize the stochastic representation of the CRM, its low-order margins, and the admissible parameter set.

First, the copula is a subset expansion. Let $g_0$ be the kernel associated with the count coordinate and $g$ the common kernel associated with each severity coordinate, both vanishing at $0$ and $1$. Each kernel arises from a calibrated pair, that is, a pair of distributions on $[0,1]$ whose mixture with weights $1-\pi$ and $\pi$ is uniform (Definition~\ref{def:calibrated}); the calibration probabilities are $\pi_N$ for the count and $\pi_X$ for the severities. Let $\bm{I}^{(k)}=(I_0,I_1,\dots,I_k)$ be a Bernoulli random vector with means $\pi_0=\pi_N$ and $\pi_r=\pi_X$ for $r\ge 1$, and define, for any subset $S\subseteq\{0,1,\dots,k\}$ with $|S|\ge 2$, the normalized centered mixed moment
\begin{equation}\label{eq:thetaS}
\theta_S = \E\left[\prod_{r\in S}\frac{I_r-\pi_r}{\pi_r}\right].
\end{equation}
Then, for each $k\ge 1$, the copula $C_{k+1}$ linking the count to the first $k$ severities is
\begin{equation}\label{eq:sarmanovSubset}
C_{k+1}(u_0,u_1,\dots,u_k) = \prod_{j=0}^k u_j + \sum_{\substack{S\subseteq\{0,\dots,k\}\\ |S|\ge 2}} \theta_S \left(\prod_{r\notin S}u_r\right) \left(\prod_{r\in S} g_r(u_r)\right),
\end{equation}
where the coordinate $u_0$ is the count and $u_1,\dots,u_k$ are the severities, and $g_r$ is $g_0$ for $r=0$ and $g$ otherwise.

Second, a low-order margin of the copula is again a copula of the same subset-expansion form, involving only the retained coordinates. Setting an omitted coordinate to $1$ eliminates every term whose subset contains it, since $g_r(1)=0$, so the margin depends only on the interaction parameters of the retained coordinates. Under exchangeability the distinct low-order parameters reduce to $\theta_{01}$ (count--severity), $\theta_{12}$ (severity--severity), and $\theta_{012}$ (count and two severities). This is precisely the structure demanded by Assumption~\ref{ass:loworder}, and it makes the stepwise procedure of Section~\ref{subsec:stepwise} feasible: Step~2 depends only $\theta_{01}$, and Step~3 additionally depends on $(\theta_{12},\theta_{012})$.

The endpoint conditions $g_r(0)=g_r(1)=0$ alone do not guarantee a valid copula. Throughout, we take the kernels from calibrated pairs and the interaction parameters from an admissible Bernoulli law, so that every displayed mixed density is nonnegative by the construction of Appendix~\ref{app:bernoulli-sarmanov}.

The Sarmanov CRM embeds \eqref{eq:sarmanovSubset} into the dependent CRM of Section~\ref{sec:background}. Let $N$ have cdf $\FN(\cdot;\eta)$ and pmf $\pN(\cdot;\eta)$, and let $X$ have cdf $\FX(\cdot;\psi)$ with density $\fX(\cdot;\psi)$. Let $U_0\sim\mathrm{Uniform}(0,1)$ be a latent uniform random variable with $N=\FN^{-1}(U_0;\eta)$, where $\FN^{-1}(u_0;\eta)=\inf\{n\in\Nzero:\FN(n;\eta)\ge u_0\}$, and set $U_j=\FX(X_j;\psi)$. For each $k\ge 1$, impose that $(U_0,U_1,\dots,U_k)$ has the copula \eqref{eq:sarmanovSubset} with exchangeable severity coordinates and parameters $(\theta_{01},\theta_{12},\theta_{012},\dots)$.

For the estimation results below, only the low-order Sarmanov margins are required. A full sequence-level Sarmanov CRM additionally requires compatibility of the Bernoulli laws across dimensions. Choosing $g(u)=u(1-u)$ and the calibrated pair $(\mathrm{Beta}(2,1),\mathrm{Beta}(1,2))$ with $\pi_X=1/2$ recovers the FGM CRM. Conditional on the Bernoulli indices, the latent uniform random variables are independent, which gives the closed forms used below and the simulation algorithms.

\subsection{Admissible dependence parameters}\label{subsec:admissible}

Estimation requires knowing which parameter values correspond to valid copulas. In the Bernoulli-mixture representation of the Sarmanov copula, the answer is a nonnegativity statement about the pmf of the latent Bernoulli random vector. For two severity coordinates, write $p_{abc}=\p(I_0=a,I_1=b,I_2=c)$ with the count index first. The full pmf has $2^3=8$ values, and exchangeability in $(I_1,I_2)$ imposes $p_{001}=p_{010}$ and $p_{101}=p_{110},$
leaving six distinct probabilities. Summation to one and the margin constraints $\p(I_0=1)=\pi_N$ and $\p(I_1=1)=\p(I_2=1)=\pi_X$ remove three further degrees of freedom, leaving three free probabilities, which determine the dependence parameters through \eqref{eq:thetaS}:
\begin{align*}
\theta_{01}&=\frac{\Cov(I_0,I_1)}{\pi_N\pi_X},\qquad
\theta_{12}=\frac{\Cov(I_1,I_2)}{\pi_X^2},\qquad
\theta_{012}=\frac{\E[(I_0-\pi_N)(I_1-\pi_X)(I_2-\pi_X)]}{\pi_N\pi_X^2}.
\end{align*}
Call a triple $(\theta_{01},\theta_{12},\theta_{012})$ trivariate-admissible (for fixed calibrated kernel pairs with calibration probabilities $\pi_N$ and $\pi_X$) if the expansion \eqref{eq:sarmanovSubset} with $k=2$ and interaction parameters $\theta_{\{0,1\}}=\theta_{\{0,2\}}=\theta_{01}$, $\theta_{\{1,2\}}=\theta_{12}$, and $\theta_{\{0,1,2\}}=\theta_{012}$ defines a valid copula on $[0,1]^3$. Let $\Theta_{B}^{(3)}$ be the set of triples arising from the moment map \eqref{eq:thetaS} from some Bernoulli law that is exchangeable in $(I_1,I_2)$.

\begin{proposition}\label{prop:admissible}
Fix calibrated kernel pairs with calibration probabilities $\pi_N$ and $\pi_X$. Every triple in $\Theta_{B}^{(3)}$ is trivariate-admissible, and the projection of $\Theta_{B}^{(3)}$ onto its first coordinate, the range of $\theta_{01}$, is the interval
$$ \frac{\max(0,\pi_N+\pi_X-1)-\pi_N\pi_X}{\pi_N\pi_X} \le \theta_{01} \le \frac{\min(\pi_N,\pi_X)-\pi_N\pi_X}{\pi_N\pi_X},$$
which is asymmetric around zero in general and equals $[-1,1]$ when $\pi_N=\pi_X=1/2$.
\end{proposition}

\begin{proof}
Given a Bernoulli law for $(I_0,I_1,I_2)$ with the stated margins, the construction of Appendix~\ref{app:bernoulli-sarmanov} produces a vector $(U_0,U_1,U_2)$ with uniform margins whose joint cdf is the subset expansion \eqref{eq:sarmanovSubset} with parameters \eqref{eq:thetaS}; a cdf with uniform margins is a copula, so the triple is trivariate-admissible. For the projection, the bivariate margin of $(I_0,I_1)$ has $p_{11}=\p(I_0=1,I_1=1)$ ranging over the Fr\'echet interval $[\max(0,\pi_N+\pi_X-1),\min(\pi_N,\pi_X)]$, and $\theta_{01}=(p_{11}-\pi_N\pi_X)/(\pi_N\pi_X)$ is increasing and affine in $p_{11}$; every value in the resulting interval is attained, for instance by the law with $(I_1,I_2)$ conditionally i.i.d.\ given $I_0$, which extends any bivariate $(I_0,I_1)$ law to an exchangeable trivariate one.
\end{proof}

For estimation, one may either optimize over $\theta_{01}$ within this interval, or parameterize the Bernoulli probabilities directly and map to $(\theta_{01},\theta_{12},\theta_{012})$, which automatically produces points of $\Theta_{B}^{(3)}$. Proposition~\ref{prop:admissible} is a sufficiency statement: every estimator in this paper is defined on the constructive set $\Theta_{B}^{(3)}$ or its bivariate projection, which is all we require.

\subsection{Closed-form margins and the observed-claim law}\label{subsec:sarm-fy}

This section derives, in closed form under conditional exchangeability, the objects the method uses. The estimators take as input the bivariate margin $h_1$ and the observed-claim law $F_Y$. From the fitted model, we then compute quantities like the moments of $Y$ and $S$. The formulas below are exact: a low-order margin retains only the interaction parameters of its own coordinates, by the marginalization property of Section~\ref{subsec:fgm2sarm}, so the higher-order interactions of the full model do not enter. In particular, $h_1$ and $F_Y$ involve only $\theta_{01}$, and the trivariate margin of Section~\ref{subsec:sarm-stepwise} involves only $(\theta_{01},\theta_{12},\theta_{012})$.

The bivariate margin of \eqref{eq:sarmanovSubset} in the count and a single severity is the bivariate Sarmanov copula
\begin{equation}\label{eq:bivSarm}
C_{2}(u_0,u_1) = u_0u_1 + \theta_{01}\, g_0(u_0)g(u_1).
\end{equation}
For the discrete coordinate, we write $v_n=\FN(n;\eta)$, $v_{n^-}=\FN(n-1;\eta)$, and the discrete kernel difference
$\Delta g_0(n;\eta)=g_0(v_n)-g_0(v_{n^-}).$
Assume the severity kernel $g:[0,1]\to\R$ is absolutely continuous with a.e.\ derivative $g'$ and $g(0)=g(1)=0$.

\begin{proposition}\label{prop:sarm-h1}
For any $n\in\Nzero$ with $\pN(n;\eta)>0$ and any $x\ge 0$, the mixed pmf--pdf of $(N,X_1)$ is
\begin{equation}\label{eq:h1sarm}
h_1(n,x;\eta,\psi,\theta_{01}) = \pN(n;\eta)\fX(x;\psi)\left[1+\theta_{01}A(n;\eta)g'(\FX(x;\psi))\right],
\end{equation}
where $g'$ is the a.e.\ derivative of $g$ and
\begin{equation}\label{eq:ABsarm}
A(n;\eta)=\frac{\Delta g_0(n;\eta)}{\pN(n;\eta)}.
\end{equation}
Consequently, the conditional density of $X_1$ given $N=n$ is
\begin{equation}\label{eq:cond1}
f_{X\mid N=n}(x;\eta,\psi,\theta_{01}) = \fX(x;\psi)\left[1+\theta_{01}A(n;\eta)g'(\FX(x;\psi))\right].
\end{equation}
\end{proposition}

\begin{proof}
Differentiating \eqref{eq:bivSarm} in $u_1$ gives $\partial_{u_1}C_2(u_0,u_1)=u_0+\theta_{01}g_0(u_0)g'(u_1)$ a.e. Differencing at $u_0=\FN(n;\eta)$ and $u_0=\FN(n-1;\eta)$ and applying \eqref{eq:mixedlik} yields \eqref{eq:h1sarm} with $A$ as in \eqref{eq:ABsarm}. That $f_{X\mid N=n}$ integrates to one follows from $\int_0^1 g'(u)\,\d u=g(1)-g(0)=0$.
\end{proof}

The conditional density \eqref{eq:cond1} involves no dependence parameter beyond $\theta_{01}$, and the count enters it only through $A(n;\eta)$. This factor is the count-coordinate analogue of the severity kernel derivative $g'(\FX(x;\psi))$: since $A(n;\eta)=\{g_0(\FN(n;\eta))-g_0(\FN(n-1;\eta))\}/\{\FN(n;\eta)-\FN(n-1;\eta)\}$, it is the mean slope of the count kernel $g_0$ across the $n$th probability interval, the difference-quotient form that a discrete count forces in place of $g_0'$. It measures the frequency-side tilt: a policy with count $n$ has its severity density perturbed in proportion to $\theta_{01}A(n;\eta)$. A policy with $n=0$ reports no claims and so contributes no observed severity; Assumption~\ref{ass:loworder} therefore states $h_1$ only for $n\ge 1$, and all estimation criteria use severity terms only for those policies $i$ with $N_i\ge 1$.

For the observed-claim law, assume $\mu_N(\eta):=\E_\eta[N]\in(0,\infty)$ and $\sum_{n\ge 1} n|\Delta g_0(n;\eta)|<\infty$, so that the series in \eqref{eq:kappaH} converges absolutely.

\begin{proposition}\label{prop:sarm-FY}
For any $x\ge 0$, the density of an arbitrary observed claim is
\begin{equation}\label{eq:fYsarm}
f_Y(x;\eta,\psi,\theta_{01}) = \fX(x;\psi)\left[1+\theta_{01}\,\kappa_N(\eta)\,g'(\FX(x;\psi))\right],
\end{equation}
where
\begin{equation}\label{eq:kappaH}
\kappa_N(\eta)=\sum_{n=1}^\infty \frac{n\pN(n;\eta)}{\mu_N(\eta)}A(n;\eta)
=\frac{\E_\eta[N\,A(N;\eta)]}{\E_\eta[N]}.
\end{equation}
\end{proposition}

\begin{proof}
Substituting \eqref{eq:cond1} into \eqref{eq:FYmixture} gives
$$F_Y(x;\eta,\psi,\theta_{01}) = \sum_{n=1}^\infty \frac{n\pN(n;\eta)}{\mu_N(\eta)} \int_0^x \fX(t;\psi)\left[1+\theta_{01}A(n;\eta)g'(\FX(t;\psi))\right]\d t. $$
Since $\sum_{n=1}^\infty n\pN(n;\eta)/\mu_N(\eta)=1$, only the interchange of sum and integral requires justification, and Fubini's theorem applies because
$$\sum_{n=1}^\infty n|\Delta g_0(n;\eta)|<\infty; \qquad \int_0^x \fX(t;\psi)|g'(\FX(t;\psi))|\d t \le \int_0^1 |g'(u)|\d u<\infty,$$
the last inequality by absolute continuity of $g$. Separating the integral and applying the substitution $u=\FX(t;\psi)$ with $g(0)=0$, which evaluates $\int_0^x \fX(t;\psi)g'(\FX(t;\psi))\,\d t$ as $g(\FX(x;\psi))$, yields the cdf
\begin{equation}\label{eq:FYsarm}
F_Y(x;\eta,\psi,\theta_{01})=\FX(x;\psi)+\theta_{01}\,\kappa_N(\eta)\,g(\FX(x;\psi)).
\end{equation}
Differentiating with respect to $x$ yields \eqref{eq:fYsarm}.
\end{proof}

The decomposition \eqref{eq:FYsarm} separates the roles of the three model components. The frequency model enters only through the scalar $\kappa_N(\eta)$. Writing \eqref{eq:kappaH} as $\kappa_N(\eta)=\E_\eta[N\,A(N;\eta)]/\E_\eta[N]$ shows it is the mean of the frequency tilt $A$ under the size-biased count law $n\pN(n;\eta)/\mu_N(\eta)$; because pooled claims are sampled size-biased by count, it is the average tilt an observed claim experiences, and we call it the size-biased frequency coefficient. It measures how strongly frequency drives the size-bias correction. The severity model sets the shape of the perturbation through $g\circ\FX$. The dependence enters only through the product $\theta_{01}\kappa_N(\eta)$. In particular, $F_Y$ involves neither $\theta_{12}$ nor $\theta_{012}$, so $(\eta,\psi,\theta_{01})$ can be estimated without reference to any higher-order dependence parameter. When $\theta_{01}\kappa_N(\eta)=0$, the correction vanishes and $F_Y=\FX$. Otherwise, depending on the sign of $\theta_{01}\kappa_N(\eta)$, the perturbation shifts mass between small and large claims while preserving total mass, since $\int_0^1 g'(u)\d u=0$.

We can now determine the moments of $Y$ and $S$.

\begin{proposition}\label{prop:sarm-mean}
Assume additionally $\E_\psi[X]<\infty$ and $\E_\psi\left[|Xg'(\FX(X))|\right]<\infty$. Then the observed-claim mean is
$$\E[Y]=\E_\psi[X]+\theta_{01}\frac{\Cov_\eta\left(N,A(N;\eta)\right)}{\E_\eta[N]}\Cov_\psi\left(X,g'(\FX(X))\right),$$
and the aggregate mean is
$$\E[S] = \E_\eta[N]\E[Y] = \E_\eta[N]\E_\psi[X]+\theta_{01}\Cov_\eta\left(N,A(N;\eta)\right)\Cov_\psi\left(X,g'(\FX(X))\right). $$
\end{proposition}

\begin{proof}
Integrating the conditional density \eqref{eq:cond1} against $x$ gives, for every $n$ with $\pN(n;\eta)>0$,
\begin{equation}\label{eq:condmeanX}
\begin{aligned}
\E[X_1\mid N=n]
&=\int_0^\infty x\fX(x;\psi)\left[1+\theta_{01}A(n;\eta)g'(\FX(x;\psi))\right]\d x\\
&=\E_\psi[X]+\theta_{01}A(n;\eta)\E_\psi[Xg'(\FX(X))].
\end{aligned}
\end{equation}
Both correction terms are covariances, because each tilt has mean zero. The substitution $u=\FX(x;\psi)$ gives
$$\E_\psi[g'(\FX(X))]=\int_0^\infty \fX(x;\psi)g'(\FX(x;\psi))\d x=\int_0^1 g'(u)\d u=g(1)-g(0)=0,$$
so $\E_\psi[Xg'(\FX(X))]=\Cov_\psi(X,g'(\FX(X)))$. Similarly, $\pN(n;\eta)A(n;\eta)=\Delta g_0(n;\eta)$ and, using a telescoping sum,
$$\E_\eta[A(N;\eta)]=\sum_{n=0}^{\infty}\Delta g_0(n;\eta)=\lim_{M\to\infty}g_0(\FN(M;\eta))=g_0(1)=0,$$
and $\E_\eta[NA(N;\eta)]=\Cov_\eta(N,A(N;\eta))$. Averaging \eqref{eq:condmeanX} over the size-biased count law of an arbitrary observed claim gives
\begin{align*}
\E[Y] &=\sum_{n=1}^\infty \frac{n\pN(n;\eta)}{\mu_N(\eta)}\E[X_1\mid N=n]=\E_\psi[X]+\theta_{01}\left(\sum_{n=1}^\infty \frac{n\pN(n;\eta)}{\mu_N(\eta)}A(n;\eta)\right)\Cov_\psi(X,g'(\FX(X))),
\end{align*}
and the sum on the right-hand side is $\kappa_N(\eta)$ as in \eqref{eq:kappaH}, which is the stated $\E[Y]$. Finally, Corollary~\ref{cor:FYgeneral} with $h(x)=x$ gives $\E[S]=\E_\eta[N]\E[Y]$. Multiplying the expression for $\E[Y]$ by $\E_\eta[N]$, we obtain the expression for $\E[S]$.
\end{proof}

The correction is a product of three co-movement factors: the frequency--severity dependence $\theta_{01}$, which by \eqref{eq:thetaS} is itself the normalized covariance $\Cov(I_0,I_1)/(\pi_N\pi_X)$ of the latent indicators; the covariance between the count and its frequency tilt $A$; and the covariance between a claim's size and its rank tilt $g'\circ\FX$. If any factor vanishes, in particular if $\theta_{01}=0$, the pooled claims are mean-unbiased for $\FX$ and $\E[S]=\E_\eta[N]\E_\psi[X]$ as under independence.

The marginal $\FX$ is not directly observable from claims because the sampling mechanism determines which claims are observed. The model-implied decomposition \eqref{eq:FYsarm}, however, expresses the observable law $F_Y$ as the true margin $\FX$ plus an explicit dependence correction. Fitting the model-implied $F_Y$ to the pooled claims is therefore the correct way to estimate the true severity margin $\FX$ in the presence of frequency--severity dependence, and the correction allows for estimation of $\FX$ even when the observed claims are not representative of the underlying severity distribution.

\section{Estimation in the Sarmanov model}\label{sec:sarmanovest}

We specialize the estimation theory of Section~\ref{sec:cl} to the Sarmanov CRM, where every quantity the estimator requires is available in closed form. We state the asymptotic result and discuss implementation.

\subsection{The composite likelihood}\label{subsec:sarm-cl}

Substituting the closed form \eqref{eq:fYsarm} into \eqref{eq:clloglik} gives the explicit objective
$$\ell_{\CL}(\eta,\psi,\theta_{01}) = \sum_{i=1}^m \log \pN(N_i;\eta) + \sum_{i=1}^m\sum_{j=1}^{N_i}\left\{\log \fX(X_{ij};\psi) + \log\left[1+\theta_{01}\kappa_N(\eta)g'(\FX(X_{ij};\psi))\right]\right\}.$$
The objective remains low-dimensional regardless of the claim counts: beyond the marginal densities, the only additional quantity is the scalar $\kappa_N(\eta)$.

Because \eqref{eq:fYsarm} involves $\theta_{01}$ only through the product $\theta_{01}\kappa_N(\eta)$, identification requires $\kappa_N(\eta_0)\neq 0$ together with identifiability of the family $(\psi,\theta_{01})\mapsto f_Y(\cdot;\eta_0,\psi,\theta_{01})$. If $\kappa_N(\eta_0)=0$, the pooled claims carry no information about $\theta_{01}$. The following specialization of Theorem~\ref{thm:CLasymp} makes these requirements explicit. Assume the data are generated from the Sarmanov CRM with $\bm{\phi}_0=(\eta_0,\psi_0,\theta_{01,0})$ in the interior of a compact parameter space $\Phi$, with $\E[N]<\infty$, strengthened to $\E[N^2]<\infty$ for asymptotic normality. We note that both moment conditions hold automatically if one assumes a Poisson distribution for the count process. Assume the count family is correctly specified and identifiable, the severity family is correctly specified, the function $A(n;\eta)$ and the kernel derivative $g'(\FX(x;\psi))$ are well defined with $\kappa_N(\eta)\in\R$ on the admissible set. Let
$$q(Z;\eta,\psi,\theta_{01}) = \log \pN(N;\eta)+\sum_{j=1}^N \log f_Y(X_j;\eta,\psi,\theta_{01})$$
and impose the standard M-estimation regularity conditions on $q$, in the sense of \citet[Chapter~5]{vanDerVaart1998}, with the cluster-level envelope and uniform-convergence requirements understood as in \citet[Chapter~19]{vanDerVaart1998}.

\begin{theorem}\label{thm:sarmCL}
Assume additionally that $\kappa_N(\eta_0)\neq 0$ and that $(\psi,\theta_{01})\mapsto f_Y(\cdot;\eta_0,\psi,\theta_{01})$ is identifiable on the admissible parameter set. Then the composite likelihood estimator satisfies $(\widehat{\eta},\widehat{\psi},\widehat{\theta}_{01})\plim (\eta_0,\psi_0,\theta_{01,0})$ and
$$\sqrt{m}
\begin{pmatrix}
\widehat{\eta}-\eta_0\\
\widehat{\psi}-\psi_0\\
\widehat{\theta}_{01}-\theta_{01,0}
\end{pmatrix}
\dto
\mathcal{N}(0,\Sigma_{\CL}),$$
with $\Sigma_{\CL}$ the policy-level Godambe sandwich of Theorem~\ref{thm:CLasymp}.
\end{theorem}

\begin{proof}
The stated conditions specialize Theorem~\ref{thm:CLasymp}, so only identification requires proof. The population criterion $Q(\bm{\phi})=\E[q(Z;\bm{\phi})]$ of Theorem~\ref{thm:CLasymp} splits into a count term $\E[\log\pN(N;\eta)]$ and a claim term $\E[\sum_{j=1}^N\log f_Y(X_j;\bm{\phi})]$. For the count term, $\E[\log\pN(N;\eta)]$ is uniquely maximized at $\eta_0$ by likelihood identifiability of the count family. For the claim term, Corollary~\ref{cor:FYgeneral} gives
$$\E\left[\sum_{j=1}^N\log f_Y(X_j;\bm{\phi})\right]=\E[N]\int_0^\infty\log f_Y(x;\bm{\phi})\d F_Y(x;\bm{\phi}_0),$$ 
which by the Kullback--Leibler inequality is maximized exactly when $f_Y(\cdot;\bm{\phi})=f_Y(\cdot;\bm{\phi}_0)$ a.e. Since \eqref{eq:fYsarm} involves $\theta_{01}$ only through $\theta_{01}\kappa_N(\eta)$, identification of $\theta_{01}$ requires $\kappa_N(\eta_0)\neq 0$, as assumed; under the assumed identifiability of $(\psi,\theta_{01})\mapsto f_Y(\cdot;\eta_0,\psi,\theta_{01})$, the claim term fixes $(\psi,\theta_{01})$ at the truth when $\eta=\eta_0$. No parameter with $\eta\neq\eta_0$ maximizes $Q$, because the count term is then strictly suboptimal while the claim term cannot exceed its value at the true density.
\end{proof}

The two identification requirements, $\kappa_N(\eta_0)\neq 0$ and identifiability of the family $(\psi,\theta_{01})\mapsto f_Y$, can be verified directly for specific Sarmanov CRMs, in particular for the Poisson--Gamma--FGM family used in the simulations.

\subsection{Implementation}\label{subsec:sarm-impl}

Every quantity in the objective is available in closed form within the Sarmanov CRM from this section, but maximizing it is not a routine likelihood fit. The dependence parameter must stay in an admissible set, the scalar $\kappa_N(\eta)$ is an infinite series, and the objective can be non-convex near the admissibility boundary. We therefore discuss practical implementation, including the choice of starting values and the computation of standard errors and diagnostics. The simulation study of Section~\ref{sec:numericalstudy} uses the choices described here.

The optimizer must first be kept inside the set of valid copulas. A Sarmanov copula is valid only when it assigns nonnegative mass to every hyperrectangle, a condition that is difficult to impose directly in high dimension. The Bernoulli-mixture representation reduces admissibility to nonnegativity of a Bernoulli pmf, a set of linear constraints. For the marginal composite likelihood, which involves only $\theta_{01}$, it suffices to optimize over the interval of Proposition~\ref{prop:admissible}. Once Step~3 or full-policy simulation is required, the triple $(\theta_{01},\theta_{12},\theta_{012})$ must lie in the trivariate Bernoulli image, and a full unbounded model must additionally satisfy the projective-coherence constraint $\theta_{12}\ge\theta_{01}^2$. The most numerically stable option is to parameterize the Bernoulli pmf directly, setting $\p(I_0=1)=\pi_N$ and $\p(I_1=1)=\pi_X$ and representing the joint probabilities $p_{ab}$ subject to $p_{ab}\ge 0$ and $\sum p_{ab}=1$. Then $\theta_{01}$ is a smooth function of linearly constrained probabilities, and the boundary of the admissible range poses no difficulty. This main insight is due to the stochastic representation that turns nonnegativity constraints on the copula into linear constraints on a Bernoulli pmf.

Evaluating the objective raises two numerical questions, one for each factor of \eqref{eq:fYsarm}. The first is the series $\kappa_N(\eta)$ and its derivatives. The definition \eqref{eq:kappaH} contains $A(n;\eta)=\Delta g_0(n;\eta)/\pN(n;\eta)$, which divides by count probabilities that become very small in the tail. Substituting $A$ cancels this denominator, since $n\pN(n;\eta)A(n;\eta)=n\Delta g_0(n;\eta)$, and leaves the numerically stable form
$$\kappa_N(\eta)=\frac{1}{\mu_N(\eta)}\sum_{n=1}^{\infty} n\,\Delta g_0(n;\eta).$$
In practice, we truncate the sum at a finite $n_{\max}$, chosen large enough that the discarded tail $\p_\eta(N>n_{\max})$ is negligible. Because $\kappa_N(\eta)$ enters $\log f_Y$, its derivatives in $\eta$ contribute to the score even within the pooled-claim term. When analytic derivatives of $\mu_N$ and $\Delta g_0$ are unavailable, automatic differentiation could suffice, and the Godambe sandwich variance remains valid as long as scores and Hessians are computed accurately.

The second is the kernel derivative $g'(\FX(x;\psi))$, which needs only the severity cdf and the kernel derivative. In calibrated-pair constructions, $g(u)=\pi(F_{[2]}(u)-F_{[1]}(u))$, so $g'(u)=\pi(f_{[2]}(u)-f_{[1]}(u))$ whenever the pair has densities, as for the classical FGM, Huang--Kotz, and related families; if a closed form is not available, we can tabulate $g'$ on a grid in $[0,1]$ and interpolate.

Even inside the admissible set, the objective can be non-convex in $\theta_{01}$ near the boundary, where $1+\theta_{01}\kappa_N(\eta)g'(\FX(x;\psi))$ approaches zero for some $x$. We therefore select starting values in stages: estimate $\eta$ from the method of moments; initialize $\psi$ by fitting the severity family to the pooled claims (which is biased but could be close to the truth); then optimize jointly over $(\psi,\theta_{01})$. To keep $\theta_{01}$ within its admissible interval throughout the optimization, one can reparameterize by $\theta_{01}=\tanh(\xi)$ and optimize over the unconstrained $\xi$, truncated to the admissible interval, so that the optimizer never leaves the valid region. Otherwise, one may use the Bernoulli pmf parameterization, which is unconstrained in the probabilities but constrained to sum to one and be nonnegative.

After the estimation step, both standard errors and diagnostics must account for the policy clustering. For standard errors, we compute the policy-level gradient $U_i(\widehat{\bm{\phi}})=\nabla q(Z_i;\widehat{\bm{\phi}})$ for each policy $i=1,\dots,m$ and average the outer products, never per-claim, for the reasons of Section~\ref{subsec:godambe}. For model adequacy, we diagnose each component at the level it targets. Randomized quantile residuals \citep{DunnSmyth1996} built from $F_Y(X_{ij};\widehat{\bm{\phi}})$, over all policies $i=1,\dots,m$ and their claims $j=1,\dots,N_i$, assess calibration of the observed-claim law, not of the latent margin $\FX$. Conditional probability integral transforms based on $\widehat F_{X\mid N=N_i}(X_{ij})$, or on the fitted conditional law of $X_2$ given $(N,X_1)$, assess the dependence structure. Because these residuals reuse estimated parameters and policy clusters are internally dependent, we calibrate them by simulating full policy clusters from the fitted model.

\subsection{Estimating the within-policy dependence}\label{subsec:sarm-stepwise}

When within-policy dependence parameters are needed, the Sarmanov structure provides the trivariate margin that Step~3 requires in closed form. Under exchangeability, the trivariate margin of \eqref{eq:sarmanovSubset} in the count and the first two severities is the copula
$$C_3(u_0,u_1,u_2) = u_0u_1u_2 +\theta_{01}g_0(u_0)\left(u_2g(u_1)+u_1g(u_2)\right) +\theta_{12}u_0\, g(u_1)g(u_2) +\theta_{012}g_0(u_0)g(u_1)g(u_2),$$
and differentiating in $(u_1,u_2)$ while differencing in $u_0$ at the mass points of $\FN$, as in \eqref{eq:mixedlik}, gives
\begin{align*}
h_2(n,x_1,x_2)
  &= \pN(n)\fX(x_1)\fX(x_2)\left[1+\theta_{01}A(n)\{g'(\FX(x_1))+g'(\FX(x_2))\}\right. \\
  &\qquad\qquad\qquad\left. +\,\theta_{12}g'(\FX(x_1))g'(\FX(x_2))+\theta_{012}A(n)g'(\FX(x_1))g'(\FX(x_2))\right].
\end{align*}
The conditional density of $X_2$ given $(N,X_1)$ is the ratio $h_2/h_1$, so Step~3 is an ordinary low-dimensional likelihood maximization. Its standard errors are based on the stacked sandwich estimator of Theorem~\ref{thm:stepwise}. We note that treating \eqref{eq:step3} as a standalone likelihood with the earlier-step parameters $(\eta,\psi,\theta_{01})$ held fixed at their estimates would ignore the uncertainty propagated from those steps and would underestimate the standard errors of $(\theta_{12},\theta_{012})$.

Higher-order within-policy dependence parameters can be estimated by continuing in the same way. One would use the margin of $(N,X_1,X_2,X_3)$ through the conditional density of $X_3$ given $(N,X_1,X_2)$, and later steps proceed analogously. As long as the relevant low-order margins stay tractable and the same M-estimation regularity framework of Theorem~\ref{thm:stepwise} applies \citep[Chapter~5]{vanDerVaart1998}, each new step appends its estimating equations to the policy-level stack. The asymptotic covariance remains a Godambe sandwich of the form \eqref{eq:SWsandwich}, with enlarged block matrices that carry the plug-in uncertainty from all earlier steps.

\section{Simulation study}\label{sec:numericalstudy}

The study answers three questions: how large the bias of naive pooled-severity fitting is and whether the composite likelihood removes it; which quantities the naive fit nonetheless estimates well, and why; and whether the policy-level Godambe standard errors and Wald tests behave as the asymptotic theory predicts, including the consequences of clustering scores at the claim level instead. A grid experiment that varies the frequency--severity dependence over a wide range addresses the first two questions. An experiment simulated from a projectively coherent Sarmanov CRM with within-policy dependence addresses the third.

\subsection{Setup}\label{subsec:sim-dgp}

Both experiments in this section use
$N\sim\mathrm{Poisson}(\lambda)$ and $X\sim\mathrm{Gamma}(\alpha,\beta),$
where $\beta$ is a rate parameter, with $(\alpha,\beta)=(2,1)$, FGM kernels $g_0(u)=g(u)=u(1-u)$, and calibration probabilities $\pi_N=\pi_X=1/2$.

We first run a grid experiment to see how the size-bias correction behaves as the frequency, the strength of dependence, and the amount of data change. We therefore vary three design factors. The Poisson mean takes values $\lambda\in\{2,5,10\}$, spanning low- to high-frequency regimes; the dependence parameter takes values $\theta_{01}\in\{-0.5,\,0,\,0.25,\,0.5,\,0.8\};$ portfolio sizes are $m\in\{500,\,2\,000,\,10\,000\}$. Each scenario is replicated $R=1\,000$ times.

We simulate directly from the Bernoulli-mixture construction, so every replication comes from a valid, projectively coherent Sarmanov CRM. For each target $\theta_{01}$ we take $\theta_{12}=|\theta_{01}|$ and $\theta_{012}=0$, a projectively coherent choice since $|\theta_{01}|\ge\theta_{01}^2$. Because $F_Y$ depends only on the bivariate count--severity margin, it varies across the grid only through $\theta_{01}$; this is the law on which we build the composite likelihood.

The grid experiment compares the two estimators. The naive margin-first estimator fits $\FX$ to the pooled claims using Gamma maximum likelihood, then estimates $\theta_{01}$ by IFM, ignoring that the pooled claims sample $F_Y$. The composite likelihood estimator maximizes the objective of Section~\ref{subsec:sarm-cl} jointly over $(\lambda,\alpha,\beta,\theta_{01})$, subject to the admissibility constraint $|\theta_{01}|\le 1$ of Proposition~\ref{prop:admissible}, enforced by imposing simple lower and upper bounds on each of $\log\alpha$, $\log\beta$, and $\theta_{01}$, with three starting values for $\theta_{01}$. The truncation level for $\kappa_N(\eta)$ is chosen so that $\p(N>n_{\max})<10^{-10}$; for $\lambda=10$ this gives $n_{\max}=36$. For the severity margin, we report the median bias of $\hat\alpha$ and $\hat\beta$, that is, the median over all replications of $\hat\alpha-\alpha$ and of $\hat\beta-\beta$, and for the portfolio mean, the median absolute relative error of $\widehat{\E}[S]$.

\subsection{The severity margin: naive bias and its correction}\label{subsec:sim-sizebias}

Figure~\ref{fig:sec7-bias} plots the median bias of the Gamma shape estimate $\hat\alpha$ (left) and rate estimate $\hat\beta$ (right) against $\theta_{01}$ at the largest portfolio size, $m=10\,000$. The naive estimator behaves exactly as predicted by Corollary~\ref{cor:naiveMLE}. Its bias does not shrink with $m$ and follows the sign of the dependence: negative dependence pushes $\hat\alpha$ down and positive dependence pushes it up, because the estimator converges to the Kullback--Leibler projection of $F_Y$, not to the true margin. The bias is strongest when counts are sparse and narrows as $\lambda$ grows, as the naive curves in Figure~\ref{fig:sec7-bias} show; the rate estimate $\hat\beta$ mirrors the pattern with the opposite sign and a smaller magnitude. Table~\ref{tab:sec7-bias-range} records the resulting bias ranges alongside the negligible bias of the composite likelihood. The composite likelihood removes all size-bias distortion from the severity margin. The residual bias also does not increase with the strength of dependence.

The mechanism behind the naive bias is the size-bias perturbation of Proposition~\ref{prop:sarm-FY}. Under positive dependence, high-count policies contribute both more claims and larger claims, so the pooled sample appears heavier-tailed than $\FX$. Fitting a Gamma distribution to the pooled sample compensates by moving $(\hat\alpha,\hat\beta)$ away from the true values. The composite likelihood instead fits $f_Y$, the distribution that the pooled sample actually follows, and recovers the marginal severity parameters from the fitted model.

\begin{figure}[ht]
\centering
\includegraphics[width=\textwidth]{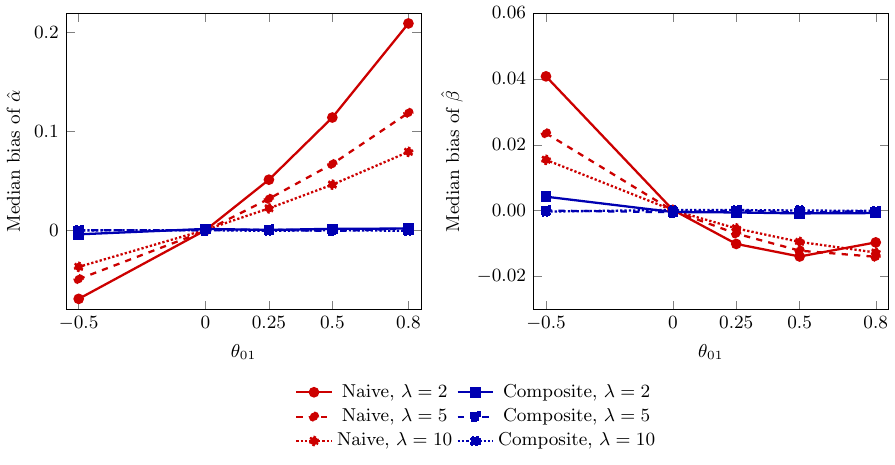}
\caption{Median bias of the Gamma shape estimate $\hat\alpha$ (left) and rate estimate $\hat\beta$ (right) at $m=10\,000$ in the grid experiment. Colour distinguishes the estimator (red: naive margin-first; blue: composite likelihood); line type distinguishes the frequency regime $\lambda\in\{2,5,10\}$.}
\label{fig:sec7-bias}
\end{figure}

\begin{table}[ht]
\centering
\caption{Grid experiment at $m=10\,000$: range over $\theta_{01}$ of the naive median bias and the largest absolute composite-likelihood median bias, for the Gamma shape $\hat\alpha$ and rate $\hat\beta$. Summaries are medians over all replications.}
\label{tab:sec7-bias-range}
\begin{tabular}{ccccc}
\toprule
 & \multicolumn{2}{c}{$\hat\alpha$} & \multicolumn{2}{c}{$\hat\beta$} \\
\cmidrule(lr){2-3}\cmidrule(lr){4-5}
$\lambda$ & Naive bias range & $\max |\text{CL bias}|$ & Naive bias range & $\max |\text{CL bias}|$ \\
\midrule
2  & $[-0.069,\,0.209]$ & 0.0040 & $[-0.014,\,0.041]$ & 0.0042 \\
5  & $[-0.049,\,0.119]$ & 0.0020 & $[-0.014,\,0.023]$ & 0.0006 \\
10 & $[-0.037,\,0.079]$ & 0.0004 & $[-0.013,\,0.015]$ & 0.0003 \\
\bottomrule
\end{tabular}
\end{table}

Although the naive severity margin is inconsistent, we show in Table~\ref{tab:sec7-mean} that both procedures estimate the aggregate mean $\E[S]$ almost identically: they are nearly identical because both estimates of $\E[S]$ are driven by the same count and pooled-claim sample means. The reason is that the pooled claims sample $F_Y$, and $\E[S]=\E[N]\,\E[Y]$, so a plug-in built from the count mean and the pooled-claim mean estimates $\E[S]$ consistently even though it misidentifies $\FX$. However, accuracy on the portfolio mean says nothing about whether the severity law has been estimated correctly. Validating a severity model through its implied $\E[S]$ can therefore accept a severely biased margin, which will have consequences for risk measures that depend on the tail of $\FX$.

Corollary~\ref{cor:naiveMLE} identifies which functionals the naive fit still estimates consistently. The naive estimator converges to the Kullback--Leibler projection $\psi^\star$ of $F_Y$ onto the working severity family, so its limiting score equations hold under $F_Y$, not under $\FX$. In the Gamma shape--rate family, the rate score fixes the fitted mean at $\alpha^\star/\beta^\star=\int x\,\d F_Y(x)=\E[Y]$. In finite samples, the MLE likewise sets the fitted mean equal to the pooled-claim mean. The naive estimate of $\E[S]$ is the product of the count sample mean and this fitted severity mean, so it converges to $\E[N]\,\E[Y]=\E[S]$. The projection matches only the score directions of the fitted family, here $x$ and $\log x$. It does not match $x^2$, tail probabilities, or the within-policy cross-products in $\E[\sum_{1\le j\ne k\le N}X_jX_k]$ that enter $\E[S^2]$ and the law of $S$. The mean is therefore available already from Steps 1 and 2, while the variance and tail risk measures may require the higher-order stepwise components. The exact match between the fitted severity mean and the pooled-claim mean is special to the exponential-family score, while the plug-in from the count and pooled-claim sample means is consistent for $\E[S]$ whatever the fitted family.

\begin{table}[ht]
\centering
\caption{Median absolute relative error of $\widehat{\E}[S]$ in the grid experiment, pooled over all $(\lambda,\theta_{01})$ scenarios and computed over all replications.}
\label{tab:sec7-mean}
\begin{tabular}{ccc}
\toprule
$m$ & Naive & Composite \\
\midrule
500    & 0.0184 & 0.0185 \\
2\,000 & 0.0093 & 0.0093 \\
10\,000 & 0.0042 & 0.0042 \\
\bottomrule
\end{tabular}
\end{table}

\subsection{Inference under a projectively coherent model}\label{subsec:sim-coherent}

The grid experiment showed that the composite likelihood estimates the correct severity margin and removes the size bias from the naive fit. This experiment investigates three things: whether the policy-level Godambe intervals attain their nominal coverage; how much clustering the scores at the claim level, rather than the policy level, costs when the severities within a policy are dependent; and whether the stepwise Wald tests for the higher-order parameters behave as the theory predicts. These require a single model with strong within-policy dependence, so we fix the coherent point $(\theta_{01},\theta_{12},\theta_{012})=(0.36,\,0.36,\,0)$, and simulate with Poisson counts $\lambda\in\{2,10\}$ and $\mathrm{Gamma}(2,1)$ severities. Here $\theta_{12}=0.36$ makes the within-policy dependence strong, which is precisely what makes clustering matter.

For each of $R=5\,000$ replications at $m\in\{2\,000,\,10\,000\}$ we fit two estimators. The $f_Y$ composite likelihood of Section~\ref{sec:cl} estimates the low-order parameters $(\lambda,\alpha,\beta,\theta_{01})$ with its policy-level Godambe sandwich; from it we read the coverage in Table~\ref{tab:sec7-coherent}. The stepwise estimator of Section~\ref{subsec:sarm-stepwise} then estimates the higher-order $(\theta_{12},\theta_{012})$, plugging in the Step-2 $h_1$ fit and using the stacked sandwich of Theorem~\ref{thm:stepwise}; from it we read the $\theta_{12}$ coverage and the two Wald tests. To isolate the clustering question, Table~\ref{tab:sec7-cluster} re-estimates $\theta_{01}$ from $h_1$ alone and sets its policy-level sandwich against a claim-level one that wrongly treats the within-policy claims as independent.

\begin{table}[ht]
\centering
\caption{Coherent-model experiment ($R=5\,000$ replications). Empirical coverage of nominal $95\%$ intervals, empirical size at level $5\%$ of the Wald test of $\theta_{012}=0$ (true null), and empirical power of the Wald test of $\theta_{12}=0$ (false null). The coverage entries for the low-order parameters $(\lambda,\alpha,\beta,\theta_{01})$ come from the $f_Y$ composite likelihood with its policy-level Godambe sandwich; the $\theta_{12}$ coverage and the two Wald tests come from the stepwise estimator with its stacked policy-level sandwich.}
\label{tab:sec7-coherent}
\begin{fitwide}
\begin{tabular}{ccccccccc}
\toprule
 & & \multicolumn{5}{c}{Coverage} & & \\
\cmidrule(lr){3-7}
$\lambda$ & $m$ & $\lambda$ & $\alpha$ & $\beta$ & $\theta_{01}$ & $\theta_{12}$ & size $T_{012}$ & power $T_{12}$ \\
\midrule
2  & 2\,000  & 0.946 & 0.927 & 0.949 & 0.940 & 0.951 & 0.048 & 0.740 \\
2  & 10\,000 & 0.949 & 0.943 & 0.955 & 0.976 & 0.954 & 0.049 & 1.000 \\
10 & 2\,000  & 0.949 & 0.951 & 0.949 & 0.963 & 0.951 & 0.049 & 1.000 \\
10 & 10\,000 & 0.951 & 0.954 & 0.949 & 0.978 & 0.947 & 0.049 & 1.000 \\
\bottomrule
\end{tabular}
\end{fitwide}
\end{table}

\begin{table}[ht]
\centering
\caption{Claim-level versus policy-level Godambe inference for $\theta_{01}$ under the coherent model ($R=5\,000$), for the Step-2 $h_1$ estimator. The $h_1$ score for $\theta_{01}$ is correlated across claims within a policy, so the claim-level sandwich understates the standard error and undercovers, increasingly so as the cluster size grows; the policy-level sandwich restores nominal coverage.}
\label{tab:sec7-cluster}
\begin{fitwide}
\begin{tabular}{ccccc}
\toprule
$\lambda$ & $m$ & cov $\theta_{01}$ (policy) & cov $\theta_{01}$ (claim) & mean SE ratio (claim/policy) \\
\midrule
2  & 2\,000  & 0.948 & 0.929 & 0.93 \\
2  & 10\,000 & 0.955 & 0.937 & 0.93 \\
10 & 2\,000  & 0.951 & 0.839 & 0.72 \\
10 & 10\,000 & 0.949 & 0.843 & 0.72 \\
\bottomrule
\end{tabular}
\end{fitwide}
\end{table}

We draw three conclusions from the experiment. First, on coverage, the Wald intervals in Table~\ref{tab:sec7-coherent} are close to nominal for every reported parameter at both portfolio sizes, ranging from $0.927$ to $0.978$; the lowest value, on the Gamma shape $\alpha$ in the sparse $\lambda=2$, $m=2\,000$ cell, is the familiar mild small-sample undercoverage of a shape parameter. Coverage for $\theta_{01}$ under the $f_Y$ composite likelihood lies between $0.940$ and $0.978$. Second, Table~\ref{tab:sec7-cluster} quantifies the cost of clustering at the wrong level when within-policy dependence is present. The $h_1$ score for $\theta_{01}$ retains within-policy correlation through $A(n;\eta)$, so the claim-level sandwich understates the standard error by about $7\%$ at $\lambda=2$ and $28\%$ at $\lambda=10$; its coverage falls to $0.84$ at $\lambda=10$, and the deficit grows with the cluster size. The policy-level sandwich restores nominal coverage; every column of Table~\ref{tab:sec7-coherent} uses the policy-level sandwich matched to its estimator. Third, the tests behave as predicted. The Wald test of the true null $\theta_{012}=0$ has empirical size between $0.048$ and $0.049$ at the nominal $5\%$ level; the test of the false null $\theta_{12}=0$ has power approaching one everywhere except the $\lambda=2$, $m=2\,000$ cell ($0.740$), where only about $1\,200$ of the $2\,000$ policies contribute a second claim to the trivariate Step-3 component; increasing $\lambda$ or $m$ resolves the deficit.

\section{Conclusion}\label{sec:conclusion}

Under frequency--severity dependence, the claims pooled across an insurance portfolio sample the observed-claim law $F_Y$, a size-biased mixture of the conditional severity distributions, and not the marginal severity law $\FX$. This single sampling fact makes margin-first procedures inconsistent for $\FX$ whenever the severity margin is fitted to uncorrected pooled claims. Having identified the true data-generating mechanism, we can develop a composite likelihood that corrects the size-bias distortion and recovers the marginal severity law. The correction is necessary whenever the pooled claims are used to fit $\FX$, and it is sufficient whenever the model implies a tractable $F_Y$. Its Godambe information is computed at the policy level, because each policy contributes a random cluster of dependent claims, and a stepwise extension estimates within-policy dependence parameters when variance or tail behaviour require them. 

The size-bias argument itself applies to any copula family with computable bivariate conditional laws. In this paper, we illustrate the method with the Sarmanov copula, which is a flexible, tractable, and projectively coherent family that can accommodate nonzero frequency--severity dependence. In this setting, every ingredient is explicit: conditional laws, the observed-claim density, aggregate moments, admissibility constraints, and a projectively coherent construction with nonzero frequency--severity dependence. This allows for efficient composite likelihood estimation and for a stepwise extension that estimates within-policy dependence parameters.

\section*{Acknowledgements}

Theorem~\ref{thm:FYlimit} emerged from conversations with Etienne Marceau on the data-generating mechanism of pooled claims in collective risk models. We acknowledge financial support from the Natural Sciences and Engineering Research Council of Canada (RGPIN-2025-06879).

\appendix
\section{Bernoulli-mixture construction}\label{app:bernoulli-sarmanov}

The Bernoulli-mixture representation of Sarmanov copulas was introduced by \citet{blier-wong2026stochastic}. We adapt it to the CRM setting, where one coordinate is a discrete count, and the rest are continuous severities. The representation serves two purposes here: it constructs the model explicitly, and it yields the admissibility statements of Proposition~\ref{prop:admissible}.

\begin{definition}[$\pi$-calibrated pair]\label{def:calibrated}
Fix $\pi\in(0,1)$. Two distribution functions $F_{[1]}$ and $F_{[2]}$ supported on $[0,1]$ are $\pi$-calibrated if
$$(1-\pi)F_{[1]}(u)+\pi F_{[2]}(u)=u,\quad u\in[0,1].$$
\end{definition}

Calibration ensures that a Bernoulli mixture of $F_{[1]}$ and $F_{[2]}$ leads to a uniform distribution. The associated kernel factor is
$g(u)=\pi(F_{[2]}(u)-F_{[1]}(u)),$ for $u\in[0,1]$,
which satisfies $g(0)=g(1)=0$ and is the perturbation kernel of the subset expansion \eqref{eq:sarmanovSubset}.

We fix $k\ge 1$ and let $\bm{I}=(I_0,I_1,\dots,I_k)\in\{0,1\}^{k+1}$ be a Bernoulli random vector with $\E[I_0]=\pi_N$ and $\E[I_j]=\pi_X$ for $j\ge 1$. Let $(F_{0,[1]},F_{0,[2]})$ be a $\pi_N$-calibrated pair for the count coordinate and $(F_{[1]},F_{[2]})$ a $\pi_X$-calibrated pair for each severity coordinate. With mutually independent components $U_{0,[r]}\sim F_{0,[r]}$ and $U_{j,[r]}\sim F_{[r]}$, $r\in\{1,2\}$, we define
$$U_0=(1-I_0)U_{0,[1]}+I_0U_{0,[2]},\qquad U_j=(1-I_j)U_{j,[1]}+I_jU_{j,[2]},\quad j=1,\dots,k.$$
Then $U_0$ and each $U_j$ are uniformly distributed on $[0,1]$, the count is $N=\FN^{-1}(U_0;\eta)$, the severities are $X_j=\FX^{-1}(U_j;\psi)$, and conditional on $\bm{I}$ the latent uniform random variables are independent, so the copula is determined by the law of $\bm{I}$ and the calibrated kernels; its cdf is exactly the subset expansion \eqref{eq:sarmanovSubset} with parameters \eqref{eq:thetaS}.

The construction proves admissibility in each fixed dimension. For an unbounded CRM, the Bernoulli random vectors across dimensions must be the finite-dimensional margins of one infinite sequence, a projective-coherence requirement. Choosing admissible Bernoulli laws separately dimension by dimension does not define a process $(X_j)_{j\ge 1}$.

Admissibility reduces to nonnegativity of the Bernoulli pmf. In the bivariate margin, with $p_{11}=\p(I_0=1,I_1=1)$, one would require
$$p_{11}\in[\max(0,\pi_N+\pi_X-1),\,\min(\pi_N,\pi_X)], \qquad \theta_{01}=\frac{p_{11}-\pi_N\pi_X}{\pi_N\pi_X},$$
which yields the interval of Proposition~\ref{prop:admissible}. For the trivariate margin under exchangeability of $(I_1,I_2)$, the six distinct atoms and the margin and sum constraints leave three free degrees of freedom; optimizing over those probabilities enforces admissibility of $(\theta_{01},\theta_{12},\theta_{012})$ automatically.

\bibliographystyle{apalike}
\bibliography{ref}

@article{blier-wong2026stochastic,
  title = {Stochastic Representation of {{Sarmanov}} Copulas},
  author = {{Blier-Wong}, Christopher},
  year = 2026,
  month = jan,
  primaryclass = {math},
  publisher = {arXiv},
  doi = {10.48550/arXiv.2601.09016},
  journal = {arXiv:2601.09016}
}

@article{blier-wong2024collective,
  title = {Collective Risk Models with {{FGM}} Dependence},
  author = {{Blier-Wong}, Christopher and Cossette, H{\'e}l{\`e}ne and Marceau, Etienne},
  year = 2025,
  journal = {Scandinavian Actuarial Journal},
  volume = {2025},
  number = {2},
  pages = {139--167},
  publisher = {Taylor \& Francis},
  issn = {0346-1238},
  doi = {10.1080/03461238.2024.2401390}
}

@article{CossetteEtAl2019,
  title = {Collective Risk Models with Dependence},
  author = {Cossette, H{\'e}l{\`e}ne and Marceau, Etienne and Mtalai, Itre},
  year = 2019,
  month = jul,
  journal = {Insurance: Mathematics and Economics},
  volume = {87},
  pages = {153--168},
  issn = {01676687},
  doi = {10.1016/j.insmatheco.2019.04.008}
}

@article{shi2020regression,
  title = {Regression for Copula-Linked Compound Distributions with Applications in Modeling Aggregate Insurance Claims},
  author = {Shi, Peng and Zhao, Zifeng},
  year = 2020,
  month = mar,
  journal = {The Annals of Applied Statistics},
  volume = {14},
  number = {1},
  pages = {357--380},
  issn = {1932-6157},
  doi = {10.1214/19-AOAS1299}
}

@article{Sklar1959,
  author = {Sklar, Abe},
  title = {Fonctions de r{\'e}partition {\`a} $n$ dimensions et leurs marges},
  journal = {Publications de l'Institut de Statistique de l'Universit{\'e} de Paris},
  volume = {8},
  pages = {229--231},
  year = {1959},
  language = {French}
}

@article{GenestNeslehova2007,
  author = {Genest, Christian and Ne{\v{s}}lehov{\'a}, Johanna},
  title = {A Primer on Copulas for Count Data},
  journal = {ASTIN Bulletin},
  volume = {37},
  number = {2},
  pages = {475--515},
  year = {2007},
  doi = {10.2143/AST.37.2.2024077}
}

@article{GenestGhoudiRivest1995,
  author = {Genest, Christian and Ghoudi, Kamel and Rivest, Louis-Paul},
  title = {A semiparametric estimation procedure of dependence parameters in multivariate families of distributions},
  journal = {Biometrika},
  volume = {82},
  number = {3},
  pages = {543--552},
  year = {1995},
  doi = {10.1093/biomet/82.3.543}
}

@article{Joe2005,
  author = {Joe, Harry},
  title = {Asymptotic efficiency of the two-stage estimation method for copula-based models},
  journal = {Journal of Multivariate Analysis},
  volume = {94},
  number = {2},
  pages = {401--419},
  year = {2005},
  doi = {10.1016/j.jmva.2004.06.003}
}

@article{Tsukahara2005,
  author = {Tsukahara, Hideatsu},
  title = {Semiparametric estimation in copula models},
  journal = {The Canadian Journal of Statistics},
  volume = {33},
  number = {3},
  pages = {357--375},
  year = {2005},
  doi = {10.1002/cjs.5540330304}
}

@article{Vardi1982,
  author = {Vardi, Yehuda},
  title = {Nonparametric estimation in the presence of length bias},
  journal = {The Annals of Statistics},
  volume = {10},
  number = {2},
  pages = {616--620},
  year = {1982},
  doi = {10.1214/aos/1176345802}
}

@article{Vardi1985,
  author = {Vardi, Yehuda},
  title = {Empirical distributions in selection bias models},
  journal = {The Annals of Statistics},
  volume = {13},
  number = {1},
  pages = {178--203},
  year = {1985},
  doi = {10.1214/aos/1176346585}
}

@article{Godambe1960,
  author = {Godambe, V. P.},
  title = {An Optimum Property of Regular Maximum Likelihood Estimation},
  journal = {The Annals of Mathematical Statistics},
  volume = {31},
  number = {4},
  pages = {1208--1211},
  year = {1960},
  doi = {10.1214/aoms/1177705693}
}

@article{GodambeHeyde1987,
  author = {Godambe, V. P. and Heyde, C. C.},
  title = {Quasi-Likelihood and Optimal Estimation},
  journal = {International Statistical Review},
  volume = {55},
  number = {3},
  pages = {231--244},
  year = {1987},
  doi = {10.2307/1403403}
}

@article{Varin2008,
  author = {Varin, Cristiano},
  title = {On composite marginal likelihoods},
  journal = {AStA Advances in Statistical Analysis},
  volume = {92},
  number = {1},
  pages = {1--28},
  year = {2008},
  doi = {10.1007/s10182-008-0060-7}
}

@article{VarinReidFirth2011,
  author = {Varin, Cristiano and Reid, Nancy and Firth, David},
  title = {An overview of composite likelihood methods},
  journal = {Statistica Sinica},
  volume = {21},
  number = {1},
  pages = {5--42},
  year = {2011},
  url = {https://www3.stat.sinica.edu.tw/statistica/j21n1/j21n11/j21n11.html}
}

@article{VarinVidoni2005,
  author = {Varin, Cristiano and Vidoni, Paolo},
  title = {A note on composite likelihood inference and model selection},
  journal = {Biometrika},
  volume = {92},
  number = {3},
  pages = {519--528},
  year = {2005},
  doi = {10.1093/biomet/92.3.519}
}

@article{DunnSmyth1996,
  author = {Dunn, Peter K. and Smyth, Gordon K.},
  title = {Randomized Quantile Residuals},
  journal = {Journal of Computational and Graphical Statistics},
  volume = {5},
  number = {3},
  pages = {236--244},
  year = {1996},
  doi = {10.1080/10618600.1996.10474708}
}

@article{LiuWang2017,
  author = {Liu, Hanming and Wang, Ruodu},
  title = {Collective risk model with dependence uncertainty},
  journal = {ASTIN Bulletin},
  volume = {47},
  number = {2},
  pages = {361--389},
  year = {2017},
  doi = {10.1017/asb.2017.4}
}

@article{CzadoEtAl2012,
  author = {Czado, Claudia and Kastenmeier, Rainer and Brechmann, Eike Christian and Min, Aleksey},
  title = {A mixed copula model for insurance claims and claim sizes},
  journal = {Scandinavian Actuarial Journal},
  year = {2012},
  volume = {2012},
  number = {4},
  pages = {278--305},
  doi = {10.1080/03461238.2010.546147}
}

@article{KramerEtAl2013,
  author = {Kr{\"a}mer, Nicole and Brechmann, Eike C. and Silvestrini, Daniel and Czado, Claudia},
  title = {Total loss estimation using copula-based regression models},
  journal = {Insurance: Mathematics and Economics},
  volume = {53},
  number = {3},
  pages = {829--839},
  year = {2013},
  doi = {10.1016/j.insmatheco.2013.09.003}
}

@article{ShiFengIvantsova2015,
  author = {Shi, Peng and Feng, Xiaoping and Ivantsova, Anastasia},
  title = {Dependent frequency--severity modeling of insurance claims},
  journal = {Insurance: Mathematics and Economics},
  volume = {64},
  pages = {417--428},
  year = {2015},
  doi = {10.1016/j.insmatheco.2015.07.006}
}

@article{GarridoGenestSchulz2016,
  author = {Garrido, Jos{\'e} and Genest, Christian and Schulz, Juliane},
  title = {Generalized linear models for dependent frequency and severity of insurance claims},
  journal = {Insurance: Mathematics and Economics},
  volume = {70},
  pages = {205--215},
  year = {2016},
  doi = {10.1016/j.insmatheco.2016.06.006}
}

@article{OhAhnLee2021,
  author = {Oh, Rosy and Ahn, Jae Youn and Lee, Woojoo},
  title = {On copula-based collective risk models: From elliptical copulas to vine copulas},
  journal = {Scandinavian Actuarial Journal},
  year = {2021},
  volume = {2021},
  number = {1},
  pages = {1--33},
  doi = {10.1080/03461238.2020.1768889}
}

@article{Cambanis1977,
  author = {Cambanis, Stamatis},
  title = {Some properties and generalizations of multivariate {Eyraud--Gumbel--Morgenstern} distributions},
  journal = {Journal of Multivariate Analysis},
  volume = {7},
  number = {4},
  pages = {551--559},
  year = {1977},
  doi = {10.1016/0047-259X(77)90066-5}
}

@article{Sarmanov1966,
  author = {Sarmanov, O. V.},
  title = {Generalized normal correlation and two-dimensional {Fr{\'e}chet} classes},
  journal = {Doklady Akademii Nauk SSSR},
  volume = {168},
  number = {1},
  pages = {32--35},
  year = {1966},
  url = {https://www.mathnet.ru/eng/dan32257},
  language = {Russian}
}

@incollection{NeweyMcFadden1994,
  author = {Newey, Whitney K. and McFadden, Daniel},
  title = {Large sample estimation and hypothesis testing},
  booktitle = {Handbook of Econometrics},
  volume = {4},
  pages = {2111--2245},
  year = {1994},
  publisher = {Elsevier},
  doi = {10.1016/S1573-4412(05)80005-4}
}

@book{vanDerVaart1998,
  author    = {van der Vaart, A. W.},
  title     = {Asymptotic Statistics},
  year      = {1998},
  publisher = {Cambridge University Press},
  address   = {Cambridge},
  series    = {Cambridge Series in Statistical and Probabilistic Mathematics}
}

@article{FreesValdez2008,
  author = {Frees, Edward W. and Valdez, Emiliano A.},
  title = {Hierarchical insurance claims modeling},
  journal = {Journal of the American Statistical Association},
  volume = {103},
  number = {484},
  pages = {1457--1469},
  year = {2008},
  doi = {10.1198/016214508000000823}
}

@article{VernicBolanceAlemany2022,
  author = {Vernic, Raluca and Bolanc{\'e}, Catalina and Alemany, Ram{\'o}n},
  title = {Sarmanov distribution for modeling dependence between the frequency and the average severity of insurance claims},
  journal = {Insurance: Mathematics and Economics},
  volume = {102},
  pages = {111--125},
  year = {2022},
  doi = {10.1016/j.insmatheco.2021.12.001}
}

@article{ShiFengBoucher2016,
  author = {Shi, Peng and Feng, Xiaoping and Boucher, Jean-Philippe},
  title = {Multilevel modeling of insurance claims using copulas},
  journal = {The Annals of Applied Statistics},
  volume = {10},
  number = {2},
  pages = {834--863},
  year = {2016},
  doi = {10.1214/16-AOAS914}
}

@article{YangShi2019,
  author = {Yang, Lu and Shi, Peng},
  title = {Multiperil rate making for property insurance using longitudinal data},
  journal = {Journal of the Royal Statistical Society: Series A},
  volume = {182},
  number = {2},
  pages = {647--668},
  year = {2019},
  doi = {10.1111/rssa.12419}
}

@book{KlugmanPanjerWillmot2012,
  author = {Klugman, Stuart A. and Panjer, Harry H. and Willmot, Gordon E.},
  title = {Loss Models: From Data to Decisions},
  edition = {4th},
  year = {2012},
  publisher = {John Wiley \& Sons},
  address = {Hoboken, NJ}
}

@book{KaasGoovaertsDhaeneDenuit2008,
  author = {Kaas, Rob and Goovaerts, Marc and Dhaene, Jan and Denuit, Michel},
  title = {Modern Actuarial Risk Theory: Using {R}},
  edition = {2nd},
  year = {2008},
  publisher = {Springer},
  address = {Berlin},
  doi = {10.1007/978-3-540-70998-5}
}

@book{DenuitDhaeneGoovaertsKaas2005,
  author = {Denuit, Michel and Dhaene, Jan and Goovaerts, Marc and Kaas, Rob},
  title = {Actuarial Theory for Dependent Risks: Measures, Orders and Models},
  year = {2005},
  publisher = {John Wiley \& Sons},
  address = {Chichester}
}

@article{PatilRao1978,
  author = {Patil, G. P. and Rao, C. R.},
  title = {Weighted distributions and size-biased sampling with applications to wildlife populations and human families},
  journal = {Biometrics},
  volume = {34},
  number = {2},
  pages = {179--189},
  year = {1978},
  doi = {10.2307/2530008}
}

@article{Rao1965,
  author = {Rao, C. R.},
  title = {On discrete distributions arising out of methods of ascertainment},
  journal = {Sankhy{\=a}: The Indian Journal of Statistics, Series A},
  volume = {27},
  number = {2/4},
  pages = {311--324},
  year = {1965}
}

@article{GillVardiWellner1988,
  author = {Gill, Richard D. and Vardi, Yehuda and Wellner, Jon A.},
  title = {Large sample theory of empirical distributions in biased sampling models},
  journal = {The Annals of Statistics},
  volume = {16},
  number = {3},
  pages = {1069--1112},
  year = {1988},
  doi = {10.1214/aos/1176350948}
}

@article{Denuit2019,
  author = {Denuit, Michel},
  title = {Size-biased transform and conditional mean risk sharing, with application to {P2P} insurance and tontines},
  journal = {ASTIN Bulletin},
  volume = {49},
  number = {3},
  pages = {591--617},
  year = {2019},
  doi = {10.1017/asb.2019.24}
}

@article{ArratiaGoldsteinKochman2019,
  author = {Arratia, Richard and Goldstein, Larry and Kochman, Fred},
  title = {Size bias for one and all},
  journal = {Probability Surveys},
  volume = {16},
  pages = {1--61},
  year = {2019},
  doi = {10.1214/13-PS221}
}

@article{NasriRemillard2023,
  title = {Identifiability and inference for copula-based semiparametric models for random vectors with arbitrary marginal distributions},
  author = {Nasri, Bouchra R. and R{\'e}millard, Bruno N.},
  year = {2023},
  month = may,
  primaryclass = {stat.ME},
  publisher = {arXiv},
  doi = {10.48550/arXiv.2301.13408},
  journal = {arXiv:2301.13408}
}

@inproceedings{Lindsay1988,
  author = {Lindsay, Bruce G.},
  title = {Composite likelihood methods},
  booktitle = {Statistical Inference from Stochastic Processes},
  editor = {Prabhu, N. U.},
  series = {Contemporary Mathematics},
  volume = {80},
  pages = {221--239},
  year = {1988},
  publisher = {American Mathematical Society},
  doi = {10.1090/conm/080/999014}
}

@article{CoxReid2004,
  author = {Cox, D. R. and Reid, Nancy},
  title = {A note on pseudolikelihood constructed from marginal densities},
  journal = {Biometrika},
  volume = {91},
  number = {3},
  pages = {729--737},
  year = {2004},
  doi = {10.1093/biomet/91.3.729}
}

@article{Bhat2014,
  author = {Bhat, Chandra R.},
  title = {The composite marginal likelihood ({CML}) inference approach with applications to discrete and mixed dependent variable models},
  journal = {Foundations and Trends in Econometrics},
  volume = {7},
  number = {1},
  pages = {1--117},
  year = {2014},
  doi = {10.1561/0800000022}
}

@article{PaceSalvanSartori2011,
  author = {Pace, Luigi and Salvan, Alessandra and Sartori, Nicola},
  title = {Adjusting composite likelihood ratio statistics},
  journal = {Statistica Sinica},
  volume = {21},
  number = {1},
  pages = {129--148},
  year = {2011}
}

@article{ChandlerBate2007,
  author = {Chandler, Richard E. and Bate, Steven},
  title = {Inference for clustered data using the independence loglikelihood},
  journal = {Biometrika},
  volume = {94},
  number = {1},
  pages = {167--183},
  year = {2007},
  doi = {10.1093/biomet/asm015}
}

@article{Eyraud1936,
  author = {Eyraud, Henri},
  title = {Les principes de la mesure des corr{\'e}lations},
  journal = {Annales de l'Universit{\'e} de Lyon, 3e s{\'e}rie, Section A, Sciences math{\'e}matiques et astronomie},
  pages = {30--47},
  year = {1936}
}

@article{Morgenstern1956,
  author = {Morgenstern, Dietrich},
  title = {{Einfache Beispiele zweidimensionaler Verteilungen}},
  journal = {Mitteilungsblatt f{\"u}r Mathematische Statistik},
  volume = {8},
  pages = {234--235},
  year = {1956}
}

@article{Gumbel1960,
  author = {Gumbel, Emil J.},
  title = {Bivariate exponential distributions},
  journal = {Journal of the American Statistical Association},
  volume = {55},
  number = {292},
  pages = {698--707},
  year = {1960},
  doi = {10.1080/01621459.1960.10483368}
}

@article{Farlie1960,
  author = {Farlie, Dennis J. G.},
  title = {The performance of some correlation coefficients for a general bivariate distribution},
  journal = {Biometrika},
  volume = {47},
  number = {3/4},
  pages = {307--323},
  year = {1960},
  doi = {10.2307/2333302}
}

@article{Lee1996,
  author = {Lee, Mei-Ling Ting},
  title = {Properties and applications of the {Sarmanov} family of bivariate distributions},
  journal = {Communications in Statistics---Theory and Methods},
  volume = {25},
  number = {6},
  pages = {1207--1222},
  year = {1996},
  doi = {10.1080/03610929608831759}
}

@article{BolanceVernic2019,
  author = {Bolanc{\'e}, Catalina and Vernic, Raluca},
  title = {Multivariate count data generalized linear models: Three approaches based on the {Sarmanov} distribution},
  journal = {Insurance: Mathematics and Economics},
  volume = {85},
  pages = {89--103},
  year = {2019},
  doi = {10.1016/j.insmatheco.2019.01.001}
}

@article{BolanceVernic2020,
  author = {Bolanc{\'e}, Catalina and Vernic, Raluca},
  title = {Frequency and severity dependence in the collective risk model: An approach based on {Sarmanov} distribution},
  journal = {Mathematics},
  volume = {8},
  number = {9},
  pages = {1400},
  year = {2020},
  doi = {10.3390/math8091400}
}

@article{AbdallahBoucherCossette2016,
  author = {Abdallah, Anas and Boucher, Jean-Philippe and Cossette, H{\'e}l{\`e}ne},
  title = {Sarmanov family of multivariate distributions for bivariate dynamic claim counts model},
  journal = {Insurance: Mathematics and Economics},
  volume = {68},
  pages = {120--133},
  year = {2016},
  doi = {10.1016/j.insmatheco.2016.01.003}
}

@article{HashorvaRatovomirija2015,
  author = {Hashorva, Enkelejd and Ratovomirija, Gildas},
  title = {On {Sarmanov} mixed {Erlang} risks in insurance applications},
  journal = {ASTIN Bulletin},
  volume = {45},
  number = {1},
  pages = {175--205},
  year = {2015},
  doi = {10.1017/asb.2014.24}
}

@book{Nelsen2006,
  author = {Nelsen, Roger B.},
  title = {An Introduction to Copulas},
  edition = {2nd},
  year = {2006},
  publisher = {Springer},
  address = {New York},
  doi = {10.1007/0-387-28678-0}
}

@article{CossetteGadouryMarceauRobert2019,
  author = {Cossette, H{\'e}l{\`e}ne and Gadoury, Simon-Pierre and Marceau, Etienne and Robert, Christian Y.},
  title = {Composite likelihood estimation method for hierarchical {Archimedean} copulas defined with multivariate compound distributions},
  journal = {Journal of Multivariate Analysis},
  volume = {172},
  pages = {59--83},
  year = {2019},
  doi = {10.1016/j.jmva.2019.01.003}
}

\end{document}